\RequirePackage{fix-cm}
\documentclass[smallextended]{svjour3}       
\smartqed  
\usepackage{graphicx}

\usepackage[english]{babel}
\usepackage{appendix}
\usepackage{enumitem}
\usepackage{url}            
\usepackage{booktabs}       
\usepackage{amsfonts}       
\usepackage{nicefrac}       
\usepackage{microtype}      
\usepackage{lipsum}
\usepackage{amssymb}
\usepackage{amsmath}
\usepackage{hyperref}       
\usepackage{xcolor}
\usepackage{comment}

\usepackage{algorithm}
\usepackage[noend]{algpseudocode}
\makeatletter
\def\BState{\State\hskip-\ALG@thistlm}
\makeatother

\usepackage{graphics}
\usepackage{subfigure}

\newcommand{\1}{\mathbf{1}}
\newcommand{\e}{\mathbf{e}}
\newcommand{\abs}[1]{\left\vert #1 \right\vert}
\newcommand{\ie}{\emph{i.e.}}
\newcommand{\cmnt}[1]{{\color{red} #1}}

\begin{document}

\title{
Herd Behaviors in Epidemics: \\ A Dynamics-Coupled Evolutionary Games Approach

}


\author{Shutian Liu         \and
       Yuhan Zhao \and
       Quanyan Zhu 
}


\institute{S. Liu, Y. Zhao and Q. Zhu \at
Department of Electrical and Computer Engineering, Tandon School of Engineering \\
New York University, Brooklyn, NY, 11201, USA \\
E-mail: \{sl6803, yhzhao, qz494\}@nyu.edu
           \and
}

\date{Received: date / Accepted: date}

\maketitle

\begin{abstract}

The recent COVID-19 pandemic has led to an increasing interest in the modeling and analysis of infectious diseases. 
The pandemic has made a significant impact on the way we behave and interact in our daily life. The past year has witnessed a strong interplay between human
behaviors and epidemic spreading.
In this paper, we propose an evolutionary game-theoretic framework to study the coupled evolutions of herd behaviors and epidemics. 
Our framework extends the classical degree-based mean-field epidemic model over complex networks by coupling it with the evolutionary game dynamics. 
The statistically equivalent individuals in a population choose their social activity intensities based on the fitness or the payoffs that depend on the state of the epidemics. Meanwhile, the  spreading of the  infectious disease over the complex network is reciprocally influenced by the players' social activities.
We analyze the coupled dynamics by studying the stationary properties of the epidemic for a given herd behavior and the structural properties of the game for a given epidemic process. The decisions of the herd turn out to be strategic substitutes.
We formulate an equivalent finite-player game and an equivalent network to represent the interactions among the finite populations.
We develop structure-preserving approximation techniques to study time-dependent properties of the joint evolution of the behavioral and epidemic dynamics.
The resemblance between the simulated coupled dynamics and the real COVID-19 statistics in the numerical experiments indicates the predictive power of our framework.

\end{abstract}

\section{Introduction}
\label{sec:intro}

The COVID-19 pandemic has unprecedentedly impacted our society in many ways. Companies, schools, and the government have shut down their offices. Many people work at home, shop online, and communicate over zoom. The past year has witnessed a litany of policies regarding social distancing, mask-wearing, and vaccination to prevent and mitigate the spreading of the pandemic. The pandemic has made a significant impact on the way we behave and interact in our daily life.
%
We have observed a strong interplay between people's behaviors and the pandemic. When the pandemic transforms the pattern of social interactions, the human behaviors also change how the infectious disease spreads.  
When the number of COVID cases goes down, reopening policies enable  more social activities to return to normal. If not done carefully, they would create second or third waves of infections, which we have witnessed recently in many countries. 

 The behaviors of individuals in the same fashion create a collective behavioral pattern that leads to the behavior of the population, which is also known as herd behavior. The herd behavior plays an important role in the pandemic. It is often driven by policies or individual incentives. For example, cities like New York and London have designed incentives for individuals to be vaccinated to reach targeted herd immunity. Many countries have enforced the policies of mask-wearing in public spaces to create herd behavior that reduces the risk of mass infection. 


Existing works on herd behaviors have focused mainly on topics related to financial markets and economics \cite{banerjee1992simple,scharfstein1990herd}.
At the same time, epidemic processes are often studied as stand-alone dynamical processes without incorporating individual behaviors  into the model \cite{pastor2015epidemic}. 
The epidemic models alone from the literature are insufficient. 
There is a need for an integrated framework that gives a holistic understanding of the pandemic together with herd behaviors.

\begin{figure}
\centering
{\includegraphics[width=0.8\textwidth]{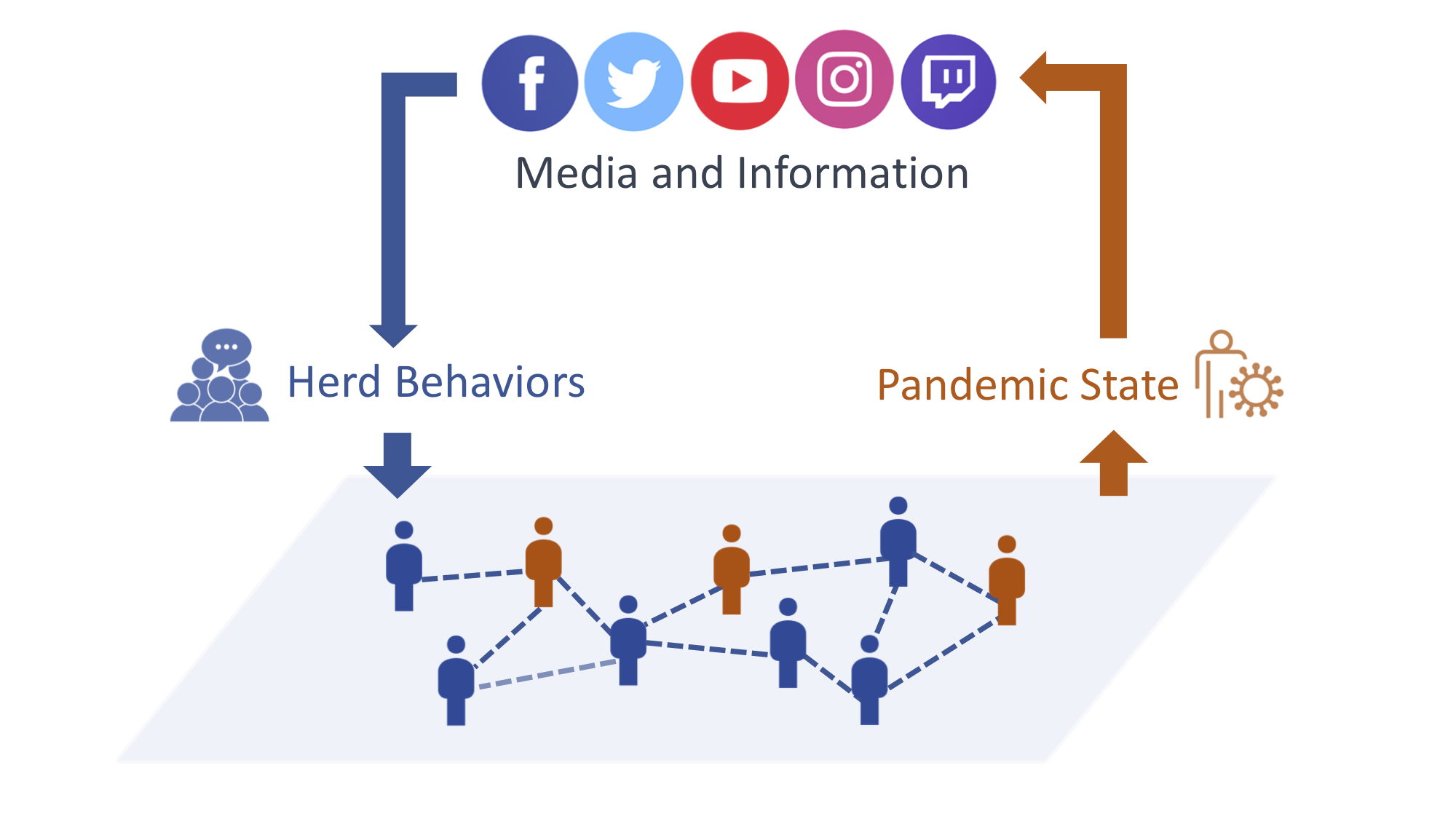}
    \label{fig:subframework2}
}
\caption[Optional caption for list of figures]{Dynamics-coupled evolutionary game framework. The herd behaviors influence epidemic spreading among the population. The media reports about the epidemic states stimulate strategy revisions, which reshape herd behaviors.
} 
\label{fig:framework}
\end{figure}

In this paper, we propose a dynamics-coupled evolutionary game-theoretic framework to model the herd behaviors that are coupled with the spreading of epidemics. 
Noncooperative games \cite{bacsar1998dynamic} are natural tools for the study of strategic decision-making of rational individuals in competitive environments.
When populations instead of a finite number of individuals are of interest, the strategy profile of the population game captures the herd behavior.
The reason lies in that the  macroscopic herd behaviors of the populations result from the microscopic strategic choices of individuals without central coordination. 

Evolutionary games study the strategic behaviors of the populations in which one population can mutate and choose a strategy against another population to maximize its fitness. The evolutionary game dynamics are population-level or mean-field dynamics that describe the evolution or the adaptive revision of the strategies when the populations interact with each other. 
The outcome of the evolutionary game and its associated dynamics is defined by the concept of evolutionarily stable strategy (ESS), which refines the concept of Nash equilibrium. 
The evolutionary dynamics provide a straightforward way to describe the macroscopic strategic interactions among the populations and the evolution of the herd behavior in response to the underlying changing environments, including the information received by the population and the fitness of the population affected by the epidemics.

One critical component of our evolutionary game framework is the modeling of the infectious disease. 
In this work, we consider a class of mean-field epidemic models over complex networks to capture the social interactions among the individuals \cite{pastor2015epidemic,pastor2001epidemic,dorogovtsev2008critical,newman2018networks}. 
The individuals over the network are assumed to be statistical equivalent within the same class of population. 
The mean-field dynamics forthrightly describe the influences of the herd behavior on the spreading of the epidemic.
We use a complex network model that is characterized by a degree distribution to represent the social interactions of the populations. 
Each individual in the network is associated with a degree or the number of connections that determines the probability of infection and thus the spreading of the disease. 


The epidemic model is consolidated into the evolutionary game framework as illustrated in Fig. \ref{fig:framework}. 
The spreading of the epidemic among the populations is affected by the social activity intensities of the individuals.
As the information and the policies concerning the epidemic being communicated to the population through public media, individuals can respond to them and adapt their social activities, constituting herd behavior at the population level. It is clear from Fig. \ref{fig:framework} that the state of the epidemics and the herd behaviors are interdependent. 


The integrated framework in Fig. \ref{fig:framework} can be mathematically described by a system of coupled differential equations. One set of differential equations represents the mean-field evolutionary dynamics of the game strategies. The other set of differential equations describes the epidemic process. It is critical to exam the structural properties of the coupled dynamics, including the stability and the steady state. 
To this end, we first discuss the stability of the epidemic dynamics under fixed herd behavior and then analyze the  structural properties of the evolutionary game under the steady states. 
We find out that, under certain conditions, there is a unique nontrivial globally asymptotically stable steady state given the herd behaviors. The players' decisions in the game turn out to be strategic substitutes. This property makes the ESS or the Nash equilibria achievable even when the individuals revise their strategies myopically on their own. This structural property is shown to hold even when the epidemic is not at its steady state.


We formulate a unified optimization problem to compute the Nash equilibrium based on an equivalent representation of the population game as a finite-player game problem, where each population is viewed as a player. Furthermore, we develop two structure-preserving approximation techniques to analyze the time-dependent evolution of the herd behaviors and the epidemics. We show that both schemes preserve the strategic substitutes property of the game.

The proposed dynamics-coupled evolutionary game provides a suitable framework to study the impact of misinformation on epidemics. We use numerical experiments to compare the simulated infection curve with the COVID-19 statistics of the infected in New York City. The prediction of two peaks in the pandemic over a time interval of interest provides a promising analytical and policy design tool for the pandemic.

The structure of this paper is as follows. Section 2 discusses related works. We introduce the general framework  in Section \ref{sec:problem setting}. 
In Section \ref{sec:long term}, we present analytical results for the case where the epidemic evolves at a faster time scale. We characterize the steady-state behavior of the epidemic, investigate the Nash equilibrium and structural properties of the game, and study the impact of misinformation. 
In Section \ref{sec:time dependent}, we extend the structural properties to different time scales and develop approximation schemes to study time-dependent behaviors.
Section \ref{sec:numerical experiments}
presents the numerical experiments. We conclude the paper in Section \ref{sec:conclusion}.

\section{Related Work}
\label{sec:related work}

The mean-field approach has been a standard tool to study the spreading of epidemics over complex networks \cite{pastor2015epidemic,pastor2001epidemic,dorogovtsev2008critical}.
The key components in this approach are the infection probabilities of the nodes, which bridge the degree distribution of the nodes with the contagion events.
This model has been used to model spreading over networks for a diverse range of applications. 
For example, in \cite{gubar2017optimal}, the authors have investigated multi-strain epidemic dynamics over complex networks to study the control policies when a single pathogen creates many strains of infections of different features. 
Another recent endeavor is \cite{chen2020optimal}, where the authors have focused on the optimal quarantining policies when multiple diseases coexist and have observed a switching phenomenon between equilibria.
Our framework extends the statistical equivalence assumptions of the standard degree-based mean-field approach.
The players in our framework are distinguished by  their degrees of connections and the strategies they choose.
The proposed coupled system of differential equations describes the flow of contagion when individuals adopt different social activity intensities. 

The connection of game theory and epidemic models has been successfully established in the celebrated work \cite{bauch2004vaccination}.
The authors have characterized the vaccination decisions in populations using the concept of convergently-stable Nash equilibrium. 
The vaccination game \cite{bauch2004vaccination} has then been used to guide policies on social distancing \cite{gosak2021endogenous} and analyze imitative behaviors in vaccination \cite{fu2011imitation}.
We are motivated by the problem settings of \cite{bauch2004vaccination} and consider a heterogeneous population over complex networks, whose connectivity is characterized by the degree distributions of individuals.
We use specific dynamic processes to model the evolution of the herd behaviors, enabling observations on the time-dependent properties of herd behaviors.

\section{Problem setting}
\label{sec:problem setting}
In this section, we describe our dynamics-coupled evolutionary game framework in detail. We first introduce the general framework and then turn to the setting under epidemic.

\subsection{The general framework}
\label{sec:problem setting:general}
\par
Consider $N\in\mathbb{N}_+$ players (nodes) over a network. Each player belongs to a subset in the set $\mathcal{D}:=\{1,2,...,D\}$ representing its degree of connectivity, \ie, a player in $d\in\mathcal{D}$ has degree $d$. 
Let $N^d\in\mathbb{N}_+$ denote the number of players who has degree $d$. 
We have $N=N^1+N^2+\cdots+N^D$. 
The degree distribution is then denoted by $[m^d]_{d\in\mathcal{D}}:=(m^1,m^2,...,m^D)\in[0,1]^D$, where $m^d=\frac{N^d}{N},\forall d\in\mathcal{D}$.
Let $\Xi$ be the finite state space of all the players and $\mathcal{S}^d:=\{s^d_1,s^d_2,...,s^d_{n^d} \} \subset [0,1]^{n^d}$ be the finite strategy space of players with degree $d$, with $|\Xi|=L$ and $|\mathcal{S}^d|=n^d$, respectively. 
Let $\mathcal{I}^d:=\{1,2,...,n^d\}$ denote the index set of the strategies in $\mathcal{S}^d, \forall d\in\mathcal{D}$. 
The strategy indexed by $i\in\mathcal{I}^d$ is $s^d_i\in\mathcal{S}^d$.
We assume throughout this paper that the strategies in each set $\mathcal{S}^d, \forall d\in\mathcal{D},$ are listed in an increasing order, \ie, for any $i,j\in\mathcal{I}^d$ such that $i>j$, we have that $s^d_i>s^d_j$ . 
Let $\mathcal{S}:=\prod_{d\in\mathcal{D}}\mathcal{S}^d$ with $|\mathcal{S}|=\sum_{d\in\mathcal{D}}n^d=n$.
Let $w^d_\xi(t)$ denote the fraction of players with degree $d$ who are in state $\xi$ at time $t$. Let $x^d_i(t)$ denote the fraction of players with degree $d$ playing strategy $s^d_i$ at time $t$. 
We use $w(t)=(w^1(t),...,w^D(t))\in \mathcal{W} \subset \mathbb{R}^{DL}_+$ and $x(t)=(x^1(t),...,x^D(t))\in \mathcal{X}=\prod_{d\in\mathcal{D}}\mathcal{X}^d\subset \mathbb{R}^{n^1 \times\cdots \times n^D}_+$ to denote the concatenations of $w^d_\xi(t)$ and $x^d_i(t)$, respectively. 
Note that we refer to $x(t)$ as the herd behavior at time $t$.
\par
Suppose that players are constantly interacting physically over the network. 
By physical interactions, we refer to the face-to-face social interactions which can cause potential changes in players' states. 
Interacting in online chat rooms is an example of social interaction which is not a physical interaction, since chatting online will not cause contagion.
We use the degree-based mean-field approach \cite{dorogovtsev2008critical,pastor2015epidemic} to capture the coupled dynamical systems on the large network. The players are assumed to be statistically equivalent if they have the same degree and the same strategy. In other words, in the large population, the players are distinguishable only based on their degree and their strategy.

Let $Q^d_\xi:\mathcal{W}\times\mathcal{X}\times[0,1]^D\rightarrow \mathbb{R}$ be a Lipschitz function describing the dynamical evolution of the fraction of players with degree $d$ who are in state $\xi$. The coupled dynamics of players' state transitions are as follows:
\begin{equation}
    \dot{w}^d_\xi (t)=Q^d_\xi \left( w(t), x(t), [m^d]_{d\in\mathcal{D}} \right), \quad  \forall d\in\mathcal{D}, \ \forall \xi\in\Xi.
    \label{eq:state dynamics}
\end{equation}

Note that in (\ref{eq:state dynamics}), the dependence on $w(t)$ emphasizes the coupling of players' state transitions, and the dependence on $[m^d]_{d\in\mathcal{D}}$ illustrates the effect of the network. We use $\Bar{w}(x)$ to denote the steady-state value of $w(x)$. The game mechanism shown in Fig. \ref{fig:subframework2} coordinates the information acquisition of the players.
As the state evolves, public media broadcasts information relevant to the states $w$ and the strategy profile $x$ to all the players at times.
We assume that the times between information broadcasts are independent, and they follow a rate $\tau$ exponential distribution. 
Each information broadcast triggers a strategic interaction where the players update their strategies based on the current information received.
For all $i,j\in\mathcal{I}^d$, let $R^d_{ij}(w(t),x(t)):\mathcal{W}\times \mathcal{X}\rightarrow [0,1]$ denote the probability of a player with degree $d$ switching from strategy $s^d_i\in\mathcal{S}^d$ to strategy $s^d_j\in\mathcal{S}^d$. 
The evolution of the fraction of players with degree $d$ playing strategy $s^d_i$ in the strategic interactions can be described as follows.
Consider a small time period $dt$. 
There will be $\tau dt$ expected information broadcasts during this period. The change of the number of players with degree $d$ playing strategy $s^d_i$ can be expressed as:
\begin{equation}
    N\left(\sum_{j\in\mathcal{I}^d}x^d_j(t)R^d_{ji}( w(t),x(t) ) -
  \sum_{j\in\mathcal{I}^d}x^d_i(t)R^d_{ij}( w(t),x(t) ) \right)\tau dt.
  \label{eq:number change of x^d_i(t)}
\end{equation}
By considering fractions of players in (\ref{eq:number change of x^d_i(t)}), we obtain the mean dynamic as follows:
\begin{equation}
\begin{aligned}
    \frac{1}{\tau}\dot{x}^d_i(t)&=\sum_{j\in\mathcal{I}^d}x^d_j(t)R^d_{ji}( w(t),x(t) )  \\
    &-
  \sum_{j\in\mathcal{I}^d}x^d_i(t)R^d_{ij}( w(t),x(t) ), \forall i\in\mathcal{I}^d, \forall d\in\mathcal{D}.
  \label{eq:strategy dynamics}
\end{aligned}
\end{equation}
Equations (\ref{eq:state dynamics}) and (\ref{eq:strategy dynamics}) constitute a system of coupled differential equations describing the joint evolution of states and strategies. 
This coupled system is depicted in Fig. \ref{fig:coupled dynamics}.

\begin{figure}
\centering
{\includegraphics[width=0.8\textwidth]{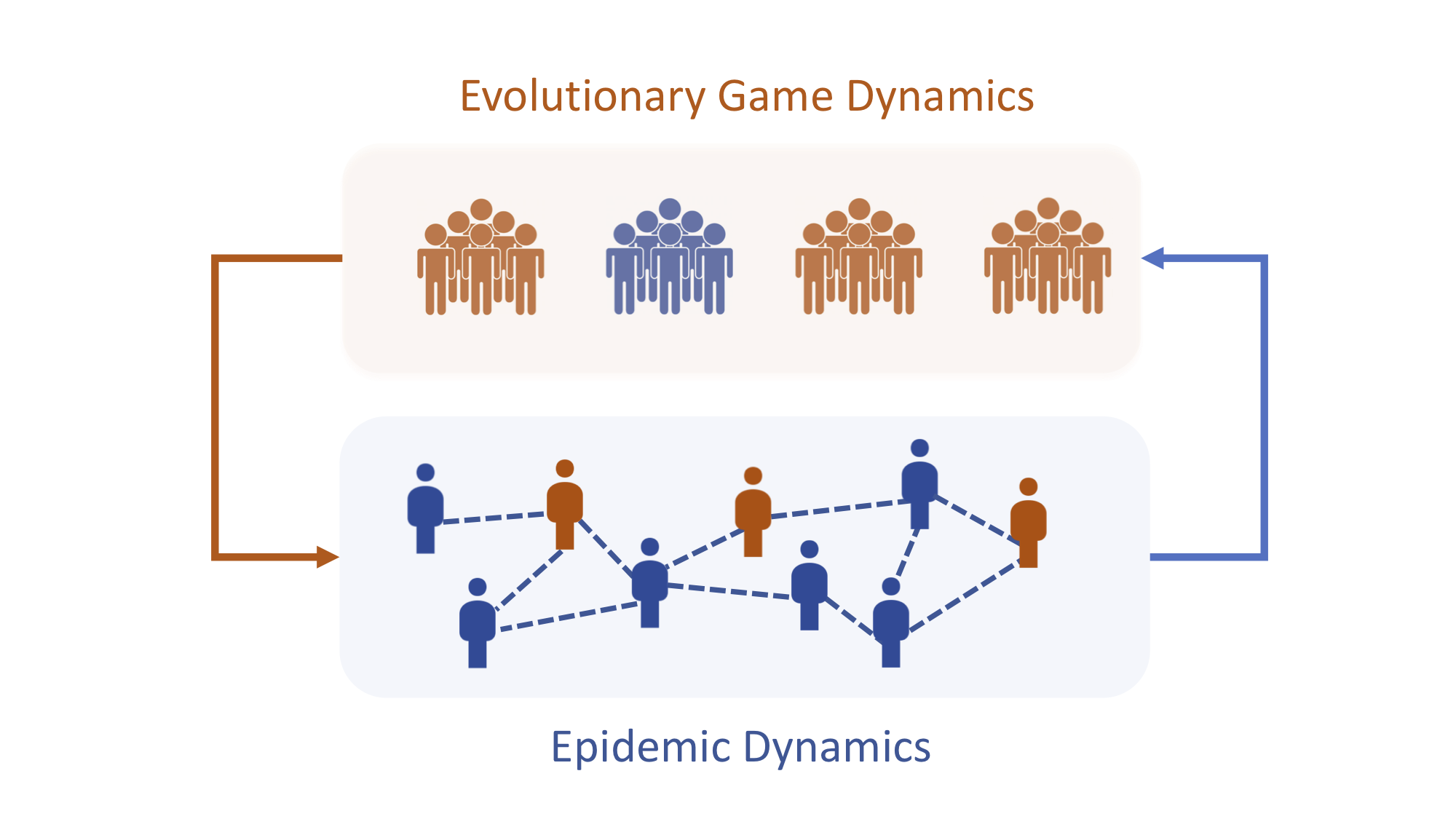}
}

\caption[Optional caption for list of figures]{Coupled evolution of game and epidemic dynamics.
} 
\label{fig:coupled dynamics}
\end{figure}

\par
Since the coupled evolutions in (\ref{eq:state dynamics}) and (\ref{eq:strategy dynamics}) involve the subpopulations of the players instead of individual players, we naturally interpret the above setting as the interactions among populations when the total number of players goes to infinity, \ie, $N\rightarrow \infty$. 
In the sequel, we refer to the players with degree $d\in\mathcal{D}$ as population $d$.
\par
For all $i\in\mathcal{I}^d$, let $F^d_i:\mathcal{W}\times \mathcal{X}\rightarrow \mathbb{R}$ denote the payoff function of a player with degree $d$ who plays strategy $s^d_i\in\mathcal{S}^d$. In general, $F^d_i$ depends on $w(t)$ and $x(t)$. The dependence on $x(t)$ characterizes the game-theoretic aspect of the framework. The dependence on $w(t)$ reveals the coupling of all players' state evolutions. 
We use $F=(F^1,...,F^D)\in\mathbb{R}^{n^1\times \cdots \times n^D}$ to denote the concatenation.
\par
Connecting with the standard definition of evolutionary games \cite{sandholm2010population}, we refer to the a herd behavior $x$ as a social state and call, with a slight abuse of terminology, $R:=(R^1,...,R^D)^T\in[0,1]^{n^1(n^1-1)\times \cdots \times  n^D(n^D-1) }$ a revision protocol, where $R^d:=(R^d_{12},...,R^d_{1n^d},...,R^d_{n^d1},...,R^d_{n^d(n^d-1)})^T\in[0,1]^{n^d (n^d-1)}, \forall d\in\mathcal{D}$.
We refer to the game defined by the payoff function $F$, the evolutionary dynamics (\ref{eq:strategy dynamics}), and the coupled state dynamics (\ref{eq:state dynamics}) as a dynamics-coupled evolutionary game.


\begin{definition}
\label{def:Nash equilibrium}
Let $NE(F)$ denote the set of NE of the game defined by the payoff $F$ and the coupled state transitions (\ref{eq:state dynamics}). A social state $x\in\mathcal{X}$ is a Nash Equilibrium (NE), i.e., $x\in NE(F)$, if for all $d\in\mathcal{D}$, $x^d\in m^d BR^d(x)$, where the set $BR^d(\cdot)$ denotes the set of best responses, i.e., $BR^d(x):=\{z\in\mathbb{R}^{n^d}_+: \1^T z=1, z_i>0, \text{ if } s^d_i\in\arg\max_{j\in\mathcal{I}^d}F^d_j(x,\Bar{w}(x))\}$.
\end{definition}
The difference between Definition \ref{def:Nash equilibrium} and the standard NE is that $x$ is the best response to the payoffs generated by $x$ together with the steady-state value of (\ref{eq:state dynamics}) given $x$. This integrates the coupling of (\ref{eq:state dynamics}) and (\ref{eq:strategy dynamics}) into the definition of NE of our framework.

\par
The terms social state and herd behavior will be used interchangeably when we refer to $x$.

\subsection{The framework under epidemics}
\label{sec:problem setting: population game}
\par
Consider the states of the players described by the Susceptible-Infected (SI) compartmental model over a degree-based network \cite{pastor2001epidemic} with degree distribution $[m^d]_{d\in\mathcal{D}}$. 
We use $I^d_i(t)$ to denote the fraction of the infected players in population $d$ adopting strategy $s^d_i\in\mathcal{S}^d$ at time $t$, $\forall i\in\mathcal{I}^d, \forall d\in\mathcal{D}$. 


The strategies $s^d_1,s^d_2,...,s^d_{n^d} \in [0,1],  d\in\mathcal{D},$ are in the action sets of an individual with degree $d$. A strategy $s^d_i$ can be interpreted as the choice of social inactivity intensity level $i$ of an individual of degree $d$. The social inactivity intensity is quantized into $n^d$ levels for an individual to choose from.
Naturally, the chosen social activity intensity (SAI) level is given by $1-s^d_i, \forall i\in\mathcal{I}^d, \forall d\in\mathcal{D}$, which can be viewed as the probability of a player behaving actively in the face-to-face physical interactions through all of her connections with other players. 
A strategy $s^d_i, i\in\mathcal{I}^d, d\in\mathcal{D},$ close to $1$ means that the player is considerably mindful when interacting socially through her connections with others.
Apart from social activity intensities, the strategies of players can also be interpreted as 
the willingness to wear masks or the probability to vaccinate.
Note that in the context of epidemics, the herd behavior describes the collective patterns in the populations' social activity intensities. 


Given a herd behavior $x\in\mathcal{X}$, the evolution of $I^d_i(t)$ is also affected by a recovery rate $\gamma\in\mathbb{R}_+$ and a contagion rate $\lambda\in\mathbb{R}_+$. 
The dynamical system, analogous to (\ref{eq:state dynamics}), describing the time evolution of $I^d_i(t)$ is
\begin{equation}
    \dot{I}^d_i(t) = -\gamma I^d_i(t) + \lambda^d_i \left( 1-I^d_i(t)\right) d \Theta(t),
    \label{eq:SI dynamics}
\end{equation}
where $\lambda^d_i=\lambda (1-s^d_i)\in\mathbb{R}_+$ denotes the activity-aware contagion rate of a player with degree $d$ and SAI $1-s^d_i$. The second term on the right-hand side of (\ref{eq:SI dynamics}) corresponds to the growth of $I^d_i(t)$.
This growth is proportional to the activity-aware contagion rate $\lambda^d_i$, the density of susceptible players $1-I^d_i(t)$, the degree of connections $d$, and the probability $\Theta(t)\in[0,1]$ that a link is connected to an infected player. This probability can be expressed as follows:
\begin{equation}
    \Theta(t)=\frac{\sum_{d\in\mathcal{D}} \left(\sum_{i\in\mathcal{I}^d}dx^d_iI^d_i(t) \right)}{\sum_{d\in\mathcal{D}}dm^d}.
    \label{eq:Theta}
\end{equation}
Since we have assumed the statistical equivalence of players with the same degree and the same strategy, 
the numerator of the right-hand-side of (\ref{eq:Theta}) consists of the sum of the probabilities that a link is connected to an infected player within each equivalence class.
The probability of a link connecting to an infected player with degree $d$ choosing strategy $s^d_i$ is proportional to $d x_{i}^d I_{i}^d(t)$. Hence, we obtain (\ref{eq:Theta}).
The consistency of (\ref{eq:SI dynamics}) and (\ref{eq:Theta}) follows from a similar reasoning as discussed in \cite{pastor2001epidemic}, since the effect of SAI has already been considered in the activity-aware contagion rate $\lambda^d_i$.

Note that (\ref{eq:Theta}) couples the dynamics in (\ref{eq:SI dynamics}) corresponding to each strategy of each population.
The concatenation of all $n$ fractions of infected is $I(t)=(I^1(t),...,I^D(t))\in[0,1]^{n^1\times \cdots \times n^D}$. 
We use $\Bar{I}^d_i$ and $\Bar{\Theta}$ to denote the steady-state quantities of $I^d_i(t)$ and $\Theta(t)$. The concatenations are $\Bar{I}$ and $\Bar{I}^d$.

Players with degree $d$ who choose strategy $s^d_i\in\mathcal{S}^d, i\in\mathcal{I}^d, d\in\mathcal{D}$ has payoff $F^d_i$, which depends on the information broadcast at the time of sampling, \ie, $I^d_i(t)$, as defined in Section \ref{sec:problem setting:general}. 
In the context of the epidemic, it takes the following form:
\begin{equation}
    F^d_i=s^d_i\mathcal{U}_{\mathrm{ina}}+(1-s^d_i)\mathcal{U}_{\mathrm{act}}^{d,i},
    \label{eq:payoff general}
\end{equation}
where $\mathcal{U}_{\mathrm{act}}^{d,i}:[0,1]\rightarrow \mathbb{R}_+$ and $\mathcal{U}_{\mathrm{ina}}\in\mathbb{R}$.
In (\ref{eq:payoff general}), $s^d_i\mathcal{U}_{\mathrm{ina}}$ represents the expected utility of being socially inactive; $(1-s^d_i)\mathcal{U}_{\mathrm{act}}^{d,i}$ represents the expected utility of being socially active.


\par
The function $\mathcal{U}_{\mathrm{act}}^{d,i}, \forall i\in\mathcal{I}^d, \forall d\in\mathcal{D}$ corresponds to the reward from getting infected through physical interactions on the network. 
Therefore, we let $\mathcal{U}_{\mathrm{act}}^{d,i}$ be decreasing in $\eta^d_i\in[0,1]$, which represents the probability that a player in population $d$ playing strategy $s^d_i$ is infected.
The probability that a player is infected can be equivalently understood as the fraction of players within the same statistical equivalent class who are infected. 
Thus, we obtain $\eta^d_i=\mathcal{O}^d_i(I^d_i(t))$, where $\mathcal{O}^d_i:[0,1] \rightarrow [0,1]$ is a player's observation of the infected fraction of players at the time of an information broadcast.
Note that the case of imperfect observations will be discussed in Section \ref{sec:long term:misinformation}. 
For now, we consider the case of perfect observations, \ie, $\eta^d_i=I^d_i(t)$. 
Since the evolution (\ref{eq:SI dynamics}) is coupled and the term (\ref{eq:Theta}) depends on the herd behavior $x$, the payoff satisfies the definition in Section \ref{sec:problem setting:general}.
Note that the rate parameter $\tau$ determines the time-scale of the coupled system of differential equations.
The sampled epidemic status is at a steady state if $\frac{1}{\tau} \rightarrow \infty$ and is time-dependent otherwise.

In this paper, we consider $\mathcal{U}_{\mathrm{act}}^{d,i}=-r_{\mathrm{act}}\eta^d_i$ for all players with reward parameter $r_{\mathrm{act}}\in\mathbb{R}_+$, for simplicity reasons.
The term $\mathcal{U}_{\mathrm{ina}}$ corresponds to isolating oneself from others. Hence, we assume that $\mathcal{U}_{\mathrm{ina}}$ is a negative constant reward for all players, \ie, $\mathcal{U}_{\mathrm{ina}}=r_{\mathrm{ina}}\in\mathbb{R}_-$. 
\par
By defining $r=\frac{ r_{\mathrm{ina}} }{r_{\mathrm{act}}} \in \mathbb{R}_-$ to be the relative reward of being socially inactive against being socially active, we obtain the payoff function suitable under (\ref{eq:SI dynamics}) as follows:
\begin{equation}
    F^d_i=s^d_ir-(1-s^d_i)\eta^d_i.
    \label{eq:payoff}
\end{equation}
Note that we have dropped the dependence of $F^d_i$ and $\eta^d_i$ on the epidemic state $I$ when we analyze equilibrium behaviors, since the epidemic state is a function that only depends on the social state $x$, as can be observed in Definition \ref{def:Nash equilibrium}. 
\par

\section{Long-term behavior}
\label{sec:long term}
In this section, we study the long-term behavior of our model under the assumption that $\frac{1}{\tau} \rightarrow \infty$. 
 The epidemic dynamics (\ref{eq:SI dynamics}) are assumed to reach the steady state more quickly than the herd behaviors. When an information broadcast changes the behavior, the epidemic dynamics would respond to it quickly and reach a steady-state before the next information broadcast arrives.

\subsection{Steady states of epidemic dynamics given social states}
\label{sec:long term:steady state analysis}
From (\ref{eq:SI dynamics}) and (\ref{eq:Theta}), we obtain the steady state as
\begin{equation}
    \Bar{I}^d_i=\frac{\theta^d_i\Bar{\Theta}}{\gamma+\theta^d_i\Bar{\Theta}},
    \label{eq:steady I of Theta}
\end{equation}
where $\theta^d_i=\lambda d(1-s^d_i)$, and
\begin{equation}
    \Bar{\Theta}=\frac{\sum_{d\in\mathcal{D}}\left( d\sum_{i\in\mathcal{I}^d}x^d_i\Bar{I}^d_i \right)}{\sum_{d\in\mathcal{D}}dm^d}.
    \label{eq:steady Theta of I}
\end{equation}
Let $\Bar{d}:=\sum_{d\in\mathcal{D}}dm^d$ denote the average degree of the network. By combining (\ref{eq:steady I of Theta}) and (\ref{eq:steady Theta of I}), we obtain the equation containing only $\Bar{\Theta}$ as follows:
\begin{equation}
    \Bar{\Theta}=\Bar{d}^{-1}\sum_{d\in\mathcal{D}} \left( d\sum_{i\in\mathcal{I}^d}\frac{x^d_i\theta^d_i\Bar{\Theta}}{\gamma+\theta^d_i\Bar{\Theta}} \right).
    \label{eq: steady Theta only}
\end{equation}
From (\ref{eq: steady Theta only}), we observe that $\Bar{\Theta}_0=0$ is always a solution. 
Accordingly, $\Bar{I}^d_{i,0}=0$ for all $d$ and for all $i$.  
We refer to $(\Bar{I}_{0},\Bar{\Theta}_0)$ as the zero steady-state pair.
At this steady state, players are uninfected no matter which statistical equivalent class they belong to; and no link leads to an infectious player.
The zero steady-state pair is often referred to as the disease-free state \cite{chen2020optimal}.
Meanwhile, there exist positive steady states, which arise from dividing $\Bar{\Theta}$ from both sides of (\ref{eq: steady Theta only}) when $\Bar{\Theta}\neq 0$:
\begin{equation}
    1=\Bar{d}^{-1}\sum_{d\in\mathcal{D}} \left( d\sum_{i\in\mathcal{I}^d}\frac{x^d_i\theta^d_i}{\gamma+\theta^d_i\Bar{\Theta}} \right).
    \label{eq:steady steate Theta}
\end{equation} 
In a positive steady-state pair $(\Bar{I}_+(x),\Bar{\Theta}_+(x))$, we have $\Bar{\Theta}_+(x)\in(0,1]$. 
It shows that a link possesses a positive probability to connect to an infected node. 
In addition, $\Bar{I}^d_{i,+}(x)=0$ if and only if $s^d_i=1$. 
It explains that a player can be safe from the epidemic only if she lives a totally isolated life. 
Note that the positive steady state pair depends on the social state $x$, since (\ref{eq: steady Theta only}) contains $x$.
The next result presents the conditions on the stability of the zero steady-state pair and the positive steady-state pair.

\begin{theorem}
\label{thm:epidemic GAS}
Consider the dynamical system in (\ref{eq:SI dynamics}) and (\ref{eq:Theta}). 
Given a social state $x\in\mathcal{X}$, the zero steady-state pair $(\Bar{I}_{0},\Bar{\Theta}_0)=(0,0)\in [0,1]^n \times [0,1]$ is globally asymptotically stable if $\frac{\lambda d (1-s^d_i)}{\gamma}<1$ for all $i\in\mathcal{I}^d$ and all $d\in\mathcal{D}$; the unique positive steady-state pair $(\Bar{I}_+(x),\Bar{\Theta}_+(x))\in [0,1]^n \times [0,1]$ is globally asymptotically stable if $\frac{\lambda d (1-s^d_i)}{\gamma}\geq1$ for all $i\in\mathcal{I}^d$ and all $d\in\mathcal{D}$.
\end{theorem}
\begin{proof}
Suppose that $\frac{\lambda d (1-s^d_i)}{\gamma}<1$ for all $i\in\mathcal{I}^d$ and for all $d\in\mathcal{D}$. Since $1-I^d_i(t)\leq1$, we obtain from (\ref{eq:SI dynamics}) that: $\frac{d}{dt}I^d_i(t)\leq-\gamma I^d_i(t)+\lambda (1-s^d_i)d\Theta(t).$
Then, it suffices to discuss the stability of the system
\begin{equation*}
    \dot{I}=\begin{pmatrix}
    -\gamma I^1_1(t) \\
    \vdots \\
    -\gamma I^D_{n^D}(t) \\
    \end{pmatrix}
    +\lambda
    \begin{pmatrix}
    1\cdot(1-s^1_1) \\
    \vdots \\
    D\cdot(1-s^D_{n^D})
    \end{pmatrix}
    \Theta(t).
\end{equation*}
Consider the Lyapunov function $V(t)=\sum_{d\in\mathcal{D}}\sum_{i\in\mathcal{I}^d}b^d_iI^d_i(t)$, where $b^d_i=\frac{dx^d_i}{\gamma}\geq 0$. The time-derivative of the Lyapunov function is:
\begin{equation*}
    \begin{aligned}
        \frac{d}{dt}V(t)&=-\sum_{d\in\mathcal{D}}\sum_{i\in\mathcal{I}^d}dx^d_iI^d_i(t)+\sum_{d\in\mathcal{D}}\sum_{i\in\mathcal{I}^d}\frac{dx^d_i\lambda (1-s^d_i)d\Theta(t)}{\gamma} \\
        & =-\Theta(t)\sum_{d\in\mathcal{D}}dm^d+\Theta(t)\sum_{d\in\mathcal{D}}\sum_{i\in\mathcal{I}^d}\frac{d^2x^d_i\lambda (1-s^d_i)}{\gamma}\\
        & =\Theta(t)\sum_{d\in\mathcal{D}}\left[ -dm^d+\frac{d^2\lambda}{\gamma}\sum_{i\in\mathcal{I}^d}x^d_i(1-s^d_i) \right].
    \end{aligned}
\end{equation*}
Combining the assumption that $\frac{\lambda d (1-s^d_i)}{\gamma}<1$ and the condition that $\sum_{i\in\mathcal{I}^d}x^d_i=m^d$, we conclude that $\frac{d}{dt}V(t)<0$ when $\Theta(t)\neq 0$. This result shows that the system $\frac{d}{dt}I^d_i(t)= -\gamma I^d_i(t)+\lambda(1-s^d_i)d\Theta(t)$ converges to $0$ as $t\rightarrow \infty$. Therefore, the system (\ref{eq:SI dynamics}) is globally asymptotically stable at the zero steady state.
\par
Suppose that the opposite condition holds, i.e. $\frac{\lambda d (1-s^d_i)}{\gamma}\geq 1$ for all $i\in\mathcal{I}^d$ and for all $d\in\mathcal{D}$. We drop the dependence on $x$ of the positive steady-state pair for simplicity.
We first show that a solution $\Bar{\Theta}_+\in(0,1]$ exists for (\ref{eq: steady Theta only}). Define $\Psi:[0,1]\rightarrow \mathbb{R}$ as:
    $\Psi(z)=\sum_{d\in\mathcal{D}} \left[ d\sum_{i\in\mathcal{I}^d}\frac{x^d_i\theta^d_i}{\gamma+\theta^d_iz} \right]$.
Since $\Psi(z)$ is a strictly decreasing function of $z$, $\Psi(0)$ achieves the maximum value and $\Psi(1)$ achieves the minimum value of $\Psi(\cdot)$. Under the condition $\frac{\lambda d (1-s^d_i)}{\gamma}\geq 1$, we obtain the inequality
\begin{equation*}
    \frac{\theta^d_i}{\gamma}\geq1\geq\frac{\theta^d_i}{\gamma+\theta^d_i}.
\end{equation*}
Multiplying by $dx^d_i$ and taking the summation over all $i\in\mathcal{I}^d$ and all $d\in\mathcal{D}$, we arrive at
\begin{equation*}
    \sum_{d\in\mathcal{D}} \left[ d\sum_{i\in\mathcal{I}^d}\frac{x^d_i\theta^d_i}{\gamma} \right]
    \geq \sum_{d\in\mathcal{D}}dm^d
    \geq
    \sum_{d\in\mathcal{D}} \left[ d\sum_{i\in\mathcal{I}^d}\frac{x^d_i\theta^d_i}{\gamma+\theta^d_i} \right],
\end{equation*}
which is equivalent to $\Psi(0)\geq\Bar{d}\geq\Psi(1)$. Hence, there exists $\Bar{\Theta}_+\in(0,1]$ such that $\Psi(\Bar{\Theta}_+)=\Bar{d}$. Moreover, $\Bar{\Theta}_+$ is unique because $\Psi(\cdot)$ is monotone. Accordingly, every element of $\Bar{I}_+$ is positive. Now, we proceed to study the stability of the positive steady-state pair $(\Bar{I}_+,\Bar{\Theta}_+)$. Define $\phi^d_i=\frac{\lambda(1-s^d_i)}{\gamma}$. Consider the following equivalent system of (\ref{eq:SI dynamics}):
\begin{equation*}
    \frac{d}{dt}I^d_i(t)=-I^d_i(t)+\phi^d_id(1-I^d_i(t))\Theta(t).
\end{equation*}
Let the density of the susceptible be $U^d_i(t)=1-I^d_i(t)$, (\ref{eq:Theta}) can be rewritten as
\begin{equation*}
\begin{aligned}
    \frac{d}{dt}\Theta(t)&=\Bar{d}^{-1}\sum_{d\in\mathcal{D}}\sum_{i\in\mathcal{I}^d}dx^d_i\frac{dI^d_i(t)}{dt}\\
    & =\Bar{d}^{-1}\sum_{d\in\mathcal{D}}\sum_{i\in\mathcal{I}^d}dx^d_i\left[ -I^d_i(t)+\phi^d_idU^d_i(t)\Theta(t) \right]\\
    & =\Theta(t)\left[ \Bar{d}^{-1}\sum_{d\in\mathcal{D}}\sum_{i\in\mathcal{I}^d}\phi^d_idU^d_i(t)dx^d_i-1 \right].\\
\end{aligned}
\end{equation*}
Consider the following Lyapunov function for the equivalent dynamical systems above:
    $V(t)=\frac{1}{2}\sum_{d\in\mathcal{D}}\sum_{i\in\mathcal{I}^d}\left[ b^d_i(U^d_i(t)-\Bar{U}^d_i)^2 \right] +\Theta(t)-\Bar{\Theta}-\Bar{\Theta}ln\frac{\Theta(t)}{\Bar{\Theta}}$,
where the parameters $b^d_i$ are defined as $b^d_i=\frac{dx^d_i}{\Bar{d}\Bar{U}^d_i}$, and the term $\Bar{U}^d_i$ denotes the steady-state quantity of $U^d_i(t)$. The time-derivative of $V(t)$ is 
\begin{equation*}
\begin{aligned}
    \frac{d}{dt}V(t)&= \sum_{d\in\mathcal{D}}\sum_{i\in\mathcal{I}^d}b^d_i(U^d_i(t)-\Bar{U}^d_i)\frac{dU^d_i(t)}{dt}+\frac{\Theta(t)-\Bar{\Theta}}{\Theta(t)}\cdot\frac{d\Theta(t)}{dt} \\
    & =\sum_{d\in\mathcal{D}}\sum_{i\in\mathcal{I}^d}b^d_i(U^d_i(t)-\Bar{U}^d_i)(I^d_i(t)-\phi^d_idU^d_i(t)\Theta(t)) \\
    & \ \ \ \ \ \ +(\Theta(t)-\Bar{\Theta})\left[ \frac{\sum_{d\in\mathcal{D}}\sum_{i\in\mathcal{I}^d}\phi^d_idU^d_i(t)dx^d_i}{\Bar{d}}-1 \right].
\end{aligned}
\end{equation*}
Since $\Bar{U}^d_i=1-\Bar{I}^d_i$ and $\Bar{d}^{-1}\sum_{d\in\mathcal{D}}\sum_{i\in\mathcal{I}^d}dx^d_i\phi^d_id\Bar{U}^d_i=1$, we obtain
\begin{equation*}
    \begin{aligned}
        \frac{d}{dt}V(t)&= 
        \sum_{d\in\mathcal{D}}\sum_{i\in\mathcal{I}^d}b^d_i(U^d_i(t)-\Bar{U}^d_i)\left[ (I^d_i(t)-\Bar{I}^d_i) -\phi^d_id(U^d_i(t)-\Bar{U}^d_i) \right] \\
        & \ \ \ \ \ \ +(\Theta(t)-\Bar{\Theta})\left[   \Bar{d}^{-1}\sum_{d\in\mathcal{D}}\sum_{i\in\mathcal{I}^d}\phi^d_idx^d_id(U^d_i(t)-\Bar{U}^d_i) \right] \\
        &= \sum_{d\in\mathcal{D}}\sum_{i\in\mathcal{I}^d}b^d_i\left[ (U^d_i(t)-\Bar{U}^d_i)(I^d_i(t)-\Bar{I}^d_i)\right] \\
        & \ \ \ \ \ \ + \Bar{d}^{-1}\sum_{d\in\mathcal{D}}\sum_{i\in\mathcal{I}^d}\frac{\phi^d_idx^d_id}{\Bar{U}^d_i}(U^d_i(t)-\Bar{U}^d_i)\left[U^d_i(t)\Theta(t)-\Bar{U}^d_i\Bar{\Theta}\right]\\
        &\ \ \ \ \ \ + \Bar{d}^{-1}\sum_{d\in\mathcal{D}}\sum_{i\in\mathcal{I}^d}\phi^d_idx^d_id \left[ (\Theta(t)-\Bar{\Theta})(U^d_i(t)-\Bar{U}^d_i) \right] \\
        &= -\sum_{d\in\mathcal{D}}\sum_{i\in\mathcal{I}^d}b^d_i(U^d_i(t)-\Bar{U}^d_i)^2\\
        &\ \ \ \ \ \ -
        \Bar{d}^{-1}\sum_{d\in\mathcal{D}}\sum_{i\in\mathcal{I}^d}\phi^d_idx^d_id \Theta(t)\left[ \frac{(U^d_i(t))^2}{\Bar{U}^d_i}-2U^d_i(t)+\Bar{U}^d_i \right].
    \end{aligned}
\end{equation*}
Since $\forall U^d_i(t)\in[0,1]$, $\frac{(U^d_i(t))^2}{\Bar{U}^d_i}-2U^d_i(t)+\Bar{U}^d_i\geq 0$, we conclude that $\frac{d}{dt}V(t)\leq 0$. Therefore, the positive steady-state pair $(\Bar{I}_+,\Bar{\Theta}_+)$ is globally asymptotically stable.

\qed
\end{proof}

\par
We focus on the positive steady-state pair $(\Bar{I}_+(x),\Bar{\Theta}_+(x))$ in the following sections, since it reveals richer properties of the herd behaviors.

\par
\subsection{Numerical computation of the steady states}
\label{sec:long term:steady state computation}
Define $M:[0,1]\rightarrow \mathbb{R}$ as $M(z)=\Bar{d}^{-1}\sum_{d\in\mathcal{D}}\left[ d\sum_{i\in\mathcal{I}^d}\frac{x^d_i\theta^d_iz}{\gamma+\theta^d_iz} \right]$.
The computation method to obtain a steady state relies on the next result.
\begin{theorem}
\label{thm:contraction}
The function $M(\cdot)$ is a contraction mapping on $[0,1]$.
\end{theorem}
\begin{proof}
Consider the component $\frac{x^d_i\theta^d_iz}{\gamma+x^d_i\theta^d_iz}$. For arbitrary $z_1,z_2\in[0,1]$, the following holds:
\begin{equation*}
\begin{aligned}
    & \quad \abs{ \frac{x^d_i\theta^d_iz_1}{\gamma+x^d_i\theta^d_iz_1}-\frac{x^d_i\theta^d_iz_2}{\gamma+x^d_i\theta^d_iz_2}} \\
    &=x^d_i\theta^d_i \abs{ \frac{\gamma(z_1-z_2)}{(\gamma+x^d_i\theta^d_iz_1)(\gamma+x^d_i\theta^d_iz_2)} } \\
    &=x^d_i\theta^d_i\beta^d_i \abs{z_1-z_2},
\end{aligned}
\end{equation*}
where $\beta^d_i=\frac{1}{(1+\gamma^{-1}x^d_i\theta^d_iz_1)(1+\gamma^{-1}x^d_i\theta^d_iz_2)}<1$. Then, summing all components, we obtain
\begin{equation*}
    \abs{M(z_1)-M(z_2)} = \frac{|z_1-z_2|}{\Bar{d}} \left( \sum_{d\in\mathcal{D}}\sum_{i\in\mathcal{I}^d}dx^d_i\theta^d_i\beta^d_i \right).
\end{equation*}
Since $\Bar{d}=\sum_{d\in\mathcal{D}}dm^d$,$m^d=\sum_{i\in\mathcal{I}^d}x^d_i$, $\theta^d_i\in[0,1]$, and $\beta^d_i\in(0,1)$, we conclude that $\Bar{d}^{-1}(\sum_{d\in\mathcal{D}}\sum_{i\in\mathcal{I}^d}dx^d_i\theta^d_i\beta^d_i)\in(0,1)$. Therefore, $M(\cdot)$ is a contraction mapping on $[0,1]$.
\qed
\end{proof}
Theorem \ref{thm:contraction} indicates that the steady state $\Bar{\Theta}$ can be obtained by the fixed-point iterations using the mapping $M(\cdot)$.

\subsection{Equilibrium analysis}
\label{sec:long term NE}
Before focusing on the NE, we first introduce an alternative interpretation of the population game in Section \ref{sec:problem setting: population game}.
Consider an equivalent $D$-player game where a player with degree $d$ plays a weighted-mixed strategy $x^d$ from the set $\mathcal{S}^d$. 
By weighted-mixed strategy, we refer to the restriction that $\1^T x^d=m^d$. Given a social state $x=(x^d,x^{-d})$, where $x^{-d}$ denotes the population states of populations other than population $d$, the expected payoff of player $d$ playing weighted-mixed strategy $x^d$ is $EF^d(x^d,x^{-d}):=(x^d)^TF^d((x^d,x^{-d}))=\sum_{i\in\mathcal{I}^d}x^d_iF^d_i((x^d,x^{-d}))$. 
Inspired by this equivalent game, we present the next result characterizing the NE of our evolutionary game.

\begin{theorem}
\label{thm:optimization problem for NE in steady state}
A social state $x^*\in\mathcal{X}$ is an NE of the game defined in Section \ref{sec:problem setting: population game} if and only if it solves the following optimization problem:
\begin{equation}
    \begin{aligned}
        \min_{x\in\mathcal{X},y\in\mathbb{R}^D} \ \ \ 
        &\sum_{d\in\mathcal{D}}-EF^d(x^d,x^{-d})+\sum_{d\in\mathcal{D}}y^dm^d \\
        \text{s.t. \  \ }
        & -F^d(x)\geq-y^d\1_{n^d}, \ \ \forall d\in\mathcal{D},\\
        &x^d\geq0, \ \ \1^Tx^d=m^d, \ \ \forall d\in\mathcal{D}.
        \label{eq:optimization prob}
    \end{aligned}
\end{equation}
\end{theorem}

\begin{proof}
The constraint $-F^d(x)\geq-y^d\1_{n^d}$ provides that $-EF^d(x)=-(x^d)^TF^d(x)\geq-(x^d)^Ty^d\1_{n^d}\geq-y^dm^d$. This implies that the objective function is nonnegative, \ie, $\sum_{d\in\mathcal{D}}-EF^d(
x)+\sum_{d\in\mathcal{D}}y^dm^d\geq0$. 

Suppose that $x^*=(x^{1*},...,x^{D*})$ is an NE of the population game. Define  $y^*=(y^{1*},...,y^{D*})$ by $y^{d*}=(m^d)^{-1}EF^d(x^{d*},x^{-d*})$ for all $d\in\mathcal{D}$. We prove that the pair $(x^*,y^*)$ is an optimal solution to problem (\ref{eq:optimization prob}) by showing that it is feasible and $\sum_{d\in\mathcal{D}}-EF^d(
x^*)+\sum_{d\in\mathcal{D}}y^{d*}m^d=0$. To prove the feasibility of $(x^d, y^*)$, it suffices to prove $-F^d(x^*)\geq-y^{d*}\1_{n^d}$. 
Since 
\begin{equation*}
    y^{d*}=(m^d)^{-1}(x^{d*})^TF^d(x^{d*},x^{-d*}),
\end{equation*}
we obtain
\begin{equation*}
    -y^{d*}\1_{n^d}=-(m^d)^{-1}(x^{d*})^TF^d(x^{d*},x^{-d*})\1_{n^d}.
\end{equation*}
Then it suffices to prove
\small
\begin{equation}
    -(x^{d*})^TF^d(x^{d*},x^{-d*})\1_{n^d}\leq
    -m^d\begin{pmatrix}
    (\e_1)^TF^d(x^{d*},x^{-d*}) \\
    \vdots \\
    (\e_{n^d})^TF^d(x^{d*},x^{-d*}) \\
    \end{pmatrix}
    \label{eq:vector ineq}
\end{equation}
\normalsize
where $\e_i$ is the vector of all zeros except for a $1$ at the $i$-th entry. 
From Definition \ref{def:Nash equilibrium}, we know that if $x^{d*}_i>0$, $s^d_i\in\arg\max_{j\in\mathcal{I}^d}F^d_j(x^*)$. 
This shows that for all $i\in\mathcal{I}^d$ such that $x^{d*}_i>0$, the values of $F^d_i(x^{d*},x^{-d*})$ are all equivalent to $\max_{j\in\mathcal{I}^d}F^d_j(x^*)$. 
Then, for all $i$ such that $x^{d*}_i>0$, since $\1^T x^{d*}=m^d$, equality holds in the $i$-th row of (\ref{eq:vector ineq}). 
For $j\in\mathcal{I}^d$ such that $x^{d*}_j=0$, inequality holds in the $j$-th row of (\ref{eq:vector ineq}) since $F^d_j(x^*)\leq\arg\max_{i\in\mathcal{I}^d}F^d_i(x^*)$. 
Hence, the pair $(x^*,y^*)$ is feasible. 
From the definition of $y^*$, we conclude that the objective function is zero under $(x^*,y^*)$. 
Therefore, $(x^*,y^*)$ solves (\ref{eq:optimization prob}). 

Suppose that $\Tilde{x}=(\Tilde{x}^1,...,\Tilde{x}^D)$ and $\Tilde{y}=(\Tilde{y}^1,...,\Tilde{y}^D)$ solve (\ref{eq:optimization prob}). Since we have found the pair $(x^*,y^*)$ under which the objective value is zero. The objective value must be zero under the pair $(\Tilde{x},\Tilde{y})$, \ie, $\sum_{d\in\mathcal{D}}-EF^d(\Tilde{x}^d,\Tilde{x}^{-d})+\sum_{d\in\mathcal{D}}\Tilde{y}^dm^d=0$. For all $x$ such that $x\geq0$ and $\1^T x^d=m^d, \forall d\in\mathcal{D}$, $-(x^d)^TF^d(\Tilde{x})\geq-\Tilde{y}^dm^d$ holds. This leads to
\begin{equation*}
\begin{aligned}
    &\sum_{d\in\mathcal{D}}-(x^d)^TF^d(\Tilde{x})
    \geq\sum_{d\in\mathcal{D}}-\Tilde{y}^dm^d \\
    =&\sum_{d\in\mathcal{D}}-EF^d(\Tilde{x}^d,\Tilde{x}^{-d})=\sum_{d\in\mathcal{D}}-(\Tilde{x^d})^TF^d(\Tilde{x}).
\end{aligned}
\end{equation*}
Since $\forall d\in\mathcal{D}$, $-(\Tilde{x}^d)^TF^d(\Tilde{x})\geq-\Tilde{y}^dm^d$ holds. Hence, $\forall d\in\mathcal{D}$, $-(\Tilde{x}^d)^TF^d(\Tilde{x})=-\Tilde{y}^dm^d$. Therefore, $\forall d\in\mathcal{D}$ and $\forall x^d$ feasible, we obtain
\begin{equation}
    -(\Tilde{x}^d)^TF^d(\Tilde{x})\leq-(x^d)^TF^d(\Tilde{x}).
    \label{eq:ineq in proof of NE and opt prob}
\end{equation}
Let $e_i$ be the vector of all zeros except for a $1$ at the $i$-th entry. Setting $x^d=e_1m^d$, $x^d=e_2m^d$, up to $x^d=e_{n^d}m^d$ in (\ref{eq:ineq in proof of NE and opt prob}), we arrive at
\begin{equation}
    (\Tilde{x}^d)^TF^d(\Tilde{x}) \geq \left( \max_{i\in\mathcal{I}^d} F^d_i(\Tilde{x}) \right) m^d.
    \label{eq:ineq becomes eq in proof of NE and opt prob}
\end{equation}
Since $\Tilde{x}\geq0$ and $\1^T \Tilde{x}^d=m^d$, equality holds in (\ref{eq:ineq becomes eq in proof of NE and opt prob}).
Thus, we conclude that for $i\in\mathcal{I}^d$ such that $\Tilde{x}^d_i>0$, $i\in\arg\max_{i\in\mathcal{I}^d}F^d_i(\Tilde{x})$. Therefore, $\Tilde{x}$ is an NE of the population game.
This completes the proof.
\qed
\end{proof}

Gradient-based algorithms can be used to numerically solve the optimization problem (\ref{eq:optimization prob}). At each iteration of the algorithm, the descent direction consists of the gradient vectors $\frac{\partial F^d_i}{\partial x}(x)$ for all $i\in\mathcal{I}^d$ and for all $d\in\mathcal{D}$. We provide below the explicit expression of the gradient vector given a social state.

With a slight abuse of notation, we specify the dependence on $x$ by writing the steady-state quantities using $\Bar{\Theta}(x)$ and $\Bar{I}^d_i(x)$. 
We express the gradient using the chain rule as $\frac{\partial F^d_i}{\partial x}(x)=-\frac{(1-s^d_i)\gamma \theta^d_i}{(\gamma+\theta^d_i\Bar{\Theta}(x))^2}\cdot\frac{\partial \Bar{\Theta}}{\partial x}(x)$. 
Next, we derive the term $\frac{\partial \Bar{\Theta}}{\partial x}(x)$ leveraging (\ref{eq:steady steate Theta}). 
Define $H:\mathbb{R}\times\mathbb{R}^n\rightarrow\mathbb{R}$ by $H(\Theta,x)=(\Bar{d})^{-1}\sum_{d\in\mathcal{D}}\left( d\sum_{i\in\mathcal{I}^d}\frac{x^d_i\theta^d_i}{\gamma+\theta^d_i\Theta} \right)-1$. 
It is obvious from the definition that $H$ is continuously differentiable with respect to both arguments. 
Suppose, given $x^*$, the pair $(\Theta^*,x^*)$ solves (\ref{eq:steady steate Theta}), \ie, $H(\Theta^*,x^*)=0$. The Jacobian of $H$ with respect to the first argument at $(\Theta^*,x^*)$ is $J_{\Theta}(\Theta^*,x^*):=\frac{\partial H}{\partial \Theta}((\Theta^*,x^*))\in\mathbb{R}$. 
From the proof of Theorem \ref{thm:epidemic GAS}, we know that $\Theta^*>0$ if $x^*$ is a social state. 
Hence, $J_{\Theta}(\Theta^*,x^*)<0$. 
Invoking the implicit function theorem, we observe that there exists a neighborhood $\mathcal{V}_{\Theta}$ of $\Theta^*$ and a neighborhood $\mathcal{V}_x$ of $x^*$, such that there is a unique continuously differentiable function $h:\mathcal{V}_x\rightarrow \mathcal{V}_{\Theta}$ satisfying $h(x^*)=\Theta^*$ and $H(h(x^*),x^*)=0$. 
Furthermore, the derivative of $h(\cdot)$ can be expressed as
\begin{equation}
    \frac{\partial h}{\partial x^d_i}(x^*)=-(J_{\Theta}(h(x^*),x^*))^{-1} \frac{\partial H}{\partial x^d_i}(h(x^*),x^*). 
    \label{eq:IFT gradient of h}
\end{equation}
Thus, the term $\frac{\partial \Bar{\Theta}}{\partial x}(x)$ can be obtained directly using (\ref{eq:IFT gradient of h}) at the given social state $x$. 
Therefore, the explicit gradient vectors are of the form
\small
\begin{equation}
    \frac{\partial F^d_i}{\partial x}(x)=\frac{(1-s^d_i)\gamma \theta^d_i}{(\gamma+\theta^d_i\Bar{\Theta}(x))^2} 
    \left( (J_{\Theta}(\Bar{\Theta}(x),x))^{-1} \frac{\partial H}{\partial x}(\Bar{\Theta}(x),x)  \right),
    \label{eq:gradient of F}
\end{equation}
\normalsize
where $\Bar{\Theta}(x)$ is obtained from the fixed-point iterations using the mapping $M(\cdot)$ and
\begin{equation}
    \frac{\partial H}{\partial x}(\Bar{\Theta}(x),x)=\frac{1}{\Bar{d}}\begin{pmatrix}
    \frac{1\cdot\theta^1_1}{\gamma+\theta^1_1 \Bar{\Theta}(x)} & \cdots & \frac{D\cdot\theta^D_{n^D}}{\gamma+\theta^D_{n^D} \Bar{\Theta}(x)}
    \end{pmatrix}.
\end{equation}
\par
In general, the optimization problem (\ref{eq:optimization prob}) is nonconvex. However, gradient-based algorithms are still applicable  to find stationary points, \ie, points with sufficiently small gradients. Moreover, we know from the proof of Theorem \ref{thm:optimization problem for NE in steady state} that the global optimal point yields a zero objective value. Therefore, we can test the stationary point obtained using gradient-based algorithms and using the objective value to determine whether it is a potential global optimal point, \ie, an NE social state.

\subsection{Long-term property of the game}
\label{sec:long term:long term property of the game }
Stability studies the structural properties of the games under which sequential plays following specific revision protocols converge to an NE. In this section, we analyze players' incentives to change their strategies when the game is played sequentially.

\par
Let $DF(x):=\frac{d}{dx}F(x)\in\mathbb{R}^{n\times n}$ denote the derivative of the payoffs with respect to the social state. From (\ref{eq:gradient of F}), we can express $DF$ as
\begin{equation}
        DF(x)= 
        \frac{\gamma}{\Bar{d}}(J_{\Theta}(\Bar{\Theta}(x),x))^{-1}
        \begin{pmatrix}
        \frac{(1-s^1_1)\theta^1_1}{(\gamma+\theta^1_1 \Bar{\Theta}(x))^2} \\
        \vdots \\
        \frac{(1-s^D_{n^D})\theta^D_{n^D}}{(\gamma+\theta^D_{n^D} \Bar{\Theta}(x))^2} \\
        \end{pmatrix}
        \begin{pmatrix}
        \frac{1\cdot \theta^1_1}{\gamma+\theta^1_1\Bar{\Theta}(x)} \\
         \vdots \\ \frac{D\cdot \theta^D_{n^D}}{\gamma+\theta^D_{n^D}\Bar{\Theta}(x)}
        \end{pmatrix}^T.
        \label{eq:DF}
\end{equation}
\normalsize
In some classes of games, such as potential games and stable games, various evolutionary dynamics show global stability. These games require special structures of the derivative matrix $DF(x)$.
Next, we investigate the structural properties of (\ref{eq:DF}).

\begin{theorem}
\label{thm:submodular game long term}
Under the assumption that every information broadcast takes place at the steady state of (\ref{eq:SI dynamics}), \ie, $\frac{1}{\tau} \rightarrow \infty$, the game defined in Section \ref{sec:problem setting: population game} is a submodular game.
\end{theorem}

\begin{proof}
Let $d$ and $c$ be two populations in set $\mathcal{D}$, $d$ and $c$ can represent the same population. Let $DF^d_c(x)\in\mathbb{R}^{n^d\times n^c}$ denote the block in (\ref{eq:DF}) corresponding to $\frac{dF^d(x)}{dx^c}$. We obtain the following:
\begin{equation*}
    DF^d_c(x)=\gamma(J_{\Theta}(\Bar{\Theta}(x),x))^{-1}\mu^d\cdot (\nu^c)^T,
\end{equation*}
where $\mu^d=\begin{pmatrix}
\frac{(1-s^d_1)\theta^d_1}{(\gamma+\theta^d_i\Bar{\Theta}(x))^2} & \cdots & \frac{(1-s^d_{n^d})\theta^d_{n^d}}{(\gamma+\theta^d_i\Bar{\Theta}(x))^2}
\end{pmatrix}^T\in\mathbb{R}^{n^d}$ and 
$\nu^c=\begin{pmatrix}
\frac{c\theta^c_1}{\gamma+\theta^c_1\Bar{\Theta}(x)}
& \cdots & 
\frac{c\theta^c_{n^c}}{\gamma+\theta^c_{n^c}\Bar{\Theta}(x)}
\end{pmatrix}^T\in\mathbb{R}^{n^c}$.
Define the matrix
\begin{equation*}
    \Sigma^d=\begin{pmatrix}
    -1 & 0 &  \cdots & 0 \\
    1 & -1 & \cdots & 0 \\
    0 & 1 & \cdots & 0 \\
    \vdots & \ddots & \ddots & \vdots \\
    0 & 0 & \cdots & 1 \\
    \end{pmatrix}
    \in\mathbb{R}^{n^d \times (n^d-1)},\ \  \forall d\in\mathcal{D}.
\end{equation*}
To prove the submodular property of the payoff functions, we need to show that the inequality $(\Sigma^d)^TDF^d_c(x)\Sigma^c\leq0$ holds for all $d,c \in\mathcal{D}$ and for all $x\in\mathcal{X}$. Ignoring positive constant terms, we obtain the following equivalent condition for all $d,c \in\mathcal{D}$,
\begin{equation*}
    \begin{pmatrix}
    \mu^d_2-\mu^d_1 \\
    \vdots \\
    \mu^d_{n^d}-\mu^d_{n^d-1} \\
    \end{pmatrix}
    \begin{pmatrix}
    \nu^c_2-\nu^c_1 & \cdots &\nu^c_{n^c}-\nu^c_{n^c-1}
    \end{pmatrix}
    \geq \mathrm{0}, 
\end{equation*}
where the dependence on $x$ is through $\mu^d$ and $\nu^c$.
A sufficient condition is
\begin{equation}
    \left(\frac{(1-s^d_{k+1})\theta^d_{k+1}}{(\gamma+\theta^d_{k+1}\Bar{\Theta}(x))^2}-\frac{(1-s^d_{k})\theta^d_{k}}{(\gamma+\theta^d_{k}\Bar{\Theta}(x))^2} \right) 
    \cdot \left( \frac{c\theta^c_{l+1}}{\gamma+\theta^c_{l+1}\Bar{\Theta}(x)}-\frac{c\theta^c_{l}}{\gamma+\theta^c_{l}\Bar{\Theta}(x)} \right) 
    \geq 0,
    \label{eq:supermodular proof ineq}
\end{equation}
for all $d,c\in\mathcal{D}$ and all $k\leq n^d-1$ and $l
\leq n^c-1$. Substituting $\theta^d_i=\lambda d (1-s^d_i)$, we observe that both $\frac{\lambda d(1-z)^2}{(\gamma+ \lambda d (1-z)\Bar{\Theta}(x))^2}$ and $\frac{c\lambda d (1-z)}{\gamma+ \lambda d (1-z) \Bar{\Theta}(x)}$ are decreasing functions of $z$ on $z\in[0,1]$. 
This proves (\ref{eq:supermodular proof ineq}). Therefore, we conclude the results. 
\qed
\end{proof}

A straightforward explanation of the above result is that the decisions in our evolutionary game are strategic substitutes; \ie, when a player with degree $d\in\mathcal{D}$ changes her strategy from $s^d_i$ to $s^d_j$ such that $s^d_i < s^d_j$, $\forall i,j\in\mathcal{I}^d$, other players, say players with degree $d'$, are more likely to choose a strategy closer to $s^{d'}_0$ from the set $\mathcal{S}^{d'}$, and vice versa.
This fact is also supported by the observations of the human behaviors under an epidemic. People tend to stay at home when the streets become crowded. There is a higher probability to get infected with a higher social interactivity. On the contrary, people tend to be outdoors if the others choose to stay at home. 

\par
The counterpart to a submodular game is a supermodular game \cite{sandholm2010population}, where decisions of players are strategic complements. In supermodular games, increases in strategies of other players result in a relatively higher strategy of a given player. This isotone property of the payoff function makes the best response correspondences of players well-behaved and the best-response dynamics with stochastic perturbation converge to perturbed NE of the game \cite{sandholm2010population}. The behavior of learning dynamics in submodular games is more involved \cite{JACKSON201595}. 
However, following \cite{topkis1979equilibrium} and \cite{dianetti2019submodular}, we obtain guarantees on the stability of certain learning processes. 
\par
Consider the best-response dynamics of the form
\begin{equation}
    x^d_{[k+1]}=\min \{m^dBR^d(x_{[k]})\}, \quad \forall d\in\mathcal{D},
    \label{eq:BRD}
\end{equation}
where the subscript $[k]$ represents iteration $k$ and $\min\{\cdot\}$ stands for choosing the least component. Let $x_{\min}^d=(m^d,0,...,0)^T$ and $x^d_{\max}=(0,...,0,m^d)^T$ denote the minimal and maximal state of population $d$. 
Let $x_{\min}=(x^1_{\min},...,x^D_{\min})$ and $x_{\max}=(x^1_{\max},...,x^D_{\max})$ denote the minimal and maximal social state. The following result \cite{topkis1979equilibrium} characterizes the stability of the learning process (\ref{eq:BRD}).
\begin{corollary}
\label{thm:BRD convergence}
There exists a minimal point $x^*_{\min}\in NE(F)$ and a maximal point $x^*_{\max}\in NE(F)$. The best-response dynamics (\ref{eq:BRD}) generate a monotonically increasing sequence which converges to $x^*_{\min}$ when the initial point is $x_{\min}$; (\ref{eq:BRD}) generates a monotonically decreasing sequence which converges to $x^*_{\max}$ when the initial point is $x_{\max}$.
\end{corollary}
\par
The results in Corollary \ref{thm:BRD convergence} have the following interpretations.
The initialization at $x_{\min}$ corresponds to the situation where players pay little attention to potential infections caused by the epidemic. 
In this scenario, players are at high SAIs and interact actively over the network. 
Through sequential revisions of strategies, players gradually become aware of the potential risks  from physical interactions and they become increasingly careful about their physical interactions with others. 
Hence, the sequence generated by (\ref{eq:BRD}) starting from $x_{\min}$ is increasing. 
The convergence to $x^*_{\min}$ shows that by naively best-responding to current payoffs, the population can eventually reach a point where no one has an incentive to further revise her strategies. 
On the contrary, the scenario where the starting point is $x_{\max}$ indicates cautious plays at the beginning, since players have no information about the potential consequences of the epidemic. Through sequential plays, players know more about the epidemic and they become more audacious, \ie, the subsequent social states generated by (\ref{eq:BRD}) after $x_{\max}$ are decreasing. And finally, there is a point where no one is willing to take more risks (e.g., going to the supermarket without wearing a mask). 
\par
Note that the maximal and the minimal points $x^*_{\max}$ and $x^*_{\min}$ do not, in general, correspond to the equilibrium points where the payoffs of players reach the maximum and the minimum, respectively \cite{dianetti2019submodular}. 
Corollary \ref{thm:BRD convergence} enables the monotone convergence to NE of the evolutionary dynamics of the form:
\begin{equation}
    \dot{x}^d=\min\{m^dBR^d(x)\}-x^d, \quad \forall d\in\mathcal{D},
    \label{eq:BRD continuous}
\end{equation}
where $\min\{\cdot\}$ selects the least element from a set.
The reason lies in the discretization of (\ref{eq:BRD continuous}): $x^d(t+\delta)=\delta \min\{m^dBR^d(x(t))\}+(1-\delta)x^d(t)$, which has the interpretation that in a small period $\delta$, only $\delta$ portion of the population revises their strategies to the one obtained using the best-responses. 
The updates (\ref{eq:BRD}) correspond to $\delta=1$.
Suppose that $t_k$ and $t_{k+1}$ are two time instances corresponding to iteration $[k]$ and $[k+1]$ in (\ref{eq:BRD}).
Since starting from $x_{\min}$, (\ref{eq:BRD}) yields $x^d_{[k]}\leq x^d_{[k+1]}$, and for any $\delta\in (0,1)$, $x^d(t_k+\delta)=\delta \min\{m^dBR^d(x(t_k))\}+(1-\delta)x^d(t_k)=\delta x^d_{[k+1]}+(1-\delta)x^d_{[k]}$.
Then, $x^d_{[k]}\leq x^d(t_k+\delta)\leq x^d_{[k+1]}$.
If we pick $\delta, \delta_1, \cdots, \delta_{end}\in [0,1]$ such that  $0<\delta_1<\delta_2<\cdots<\delta_{end}<1$, the same relation follows: $x^d_{[k]}\leq x^d(t_k+\delta_1)\leq \cdots \leq x^d(t_k+\delta_{end}) \leq x^d_{[k+1]}$. 
Therefore, the discretization of (\ref{eq:BRD continuous}) is monotone between $t_{k}$ and $t_{k+1}$ for arbitrary choices of an increasing sequence of $\delta$. This shows the monotonicity of (\ref{eq:BRD continuous}) and its convergence to the NE from $x_{\min}$, when we let $\delta \rightarrow 0$. The scenario where the starting point is $x_{\max}$ follows the similar reasoning.

\subsection{Misinformation broadcasting}
\label{sec:long term:misinformation}
Information plays an important role in shaping human behaviors. 
In the information broadcast, the media can control the strategies of players through the design of the information.
In this subsection, we investigate manipulations on players' observations of the status of the epidemic. 
For simplicity reasons, we assume that the information broadcast only contains $\bar\Theta$, which represents the average probability that a link on the network connects to an infected node, at a given social state.
This assumption can be understood as the total infected number of people  reported by the news. 
In addition, we assume that the strategy sets $\mathcal{S}^d$ are identical for all $d\in\mathcal{D}$ with the minimal element denoted by $s_{\min}$ and the maximal element denoted by $s_{\max}$. 
\par
Define $\mathcal{F}^d(s):[s_{\min},s_{\max}]\rightarrow \mathbb{R}$ by $ \mathcal{F}^d(s)=sr-(1-s)\eta^d_i=sr-(1-s)\mathcal{O}^d_i(\bar{I}^d_i)$ for all $d\in\mathcal{D}$.
The function $\mathcal{F}^d(s)$ extends the payoff of players with degree $d\in\mathcal{D}$ to a continuous function defined on the continuum $[s_{\min},s_{\max}]$.
This extension helps analyze the properties of the payoff when the strategies are perturbed.



Combining (\ref{eq:payoff}) and (\ref{eq:steady I of Theta}), we obtain
\begin{equation*}
    \mathcal{F}^d(s)=sr-(1-s)\frac{\lambda d (1-s) \bar{\Theta}}{\gamma+\lambda d (1-s)\bar{\Theta}}.
\end{equation*}
The derivative of $\mathcal{F}^d$ is 
\begin{equation}
    \frac{d}{ds}\mathcal{F}^d(s)=r+\frac{2\bar{\Theta}\frac{\lambda d (1-s)}{\gamma}+\bar{\Theta}^2(\frac{\lambda d (1-s)}{\gamma})^2}{(1+\bar{\Theta}\frac{\lambda d (1-s)}{\gamma})^2}.
    \label{eq:derivative  dF(s)/ds }
\end{equation}
The existence of dominant strategies depends on the sign of $\frac{d}{ds}\mathcal{F}^d(s)$. 
We observe that $\frac{d}{ds}\mathcal{F}^d(s)$ is a strictly increasing function of $\Bar{\Theta}$ under the condition that $\frac{\lambda d(1-s^d_i)}{\gamma}\geq 1$ for all $i\in\mathcal{I}^d$ and all $d\in\mathcal{D}$.
Therefore, the smallest value of $\frac{d}{ds}\mathcal{F}^d(s)$ appears when $\Bar{\Theta}$ approaches $0$ and the largest appears when $\Bar{\Theta}$ approaches $1$. The following result presents the conditions on the value of the relative reward $r$ for achieving dominant strategies.

\begin{theorem}
\label{thm:dominant strategy}
Under the assumption that $\frac{\lambda d(1-s^d_i)}{\gamma}\geq 1$ for all $i\in\mathcal{I}^d$ and all $d\in\mathcal{D}$, $s_{\min}$ is dominant for all $d\in\mathcal{D}$ if the following inequality holds:
\begin{equation}
    \abs{r} \geq 1-\frac{1}{\left( 1+\frac{\lambda d (1-s_{\min})}{\gamma}\right)^2}, \quad \forall d\in\mathcal{D}.
    \label{eq:inequality for dominant strategy}
\end{equation}
\end{theorem}
\begin{proof}
From (\ref{eq:derivative  dF(s)/ds }), we know that the minimal value of the second term on the right-hand side is $0$ when $\Bar{\Theta}=0$. In this case, $\frac{d}{ds}\mathcal{F}^d(s)\leq0$ since $r$ is negative. Therefore, dominant strategies for all players can only appear if $\frac{d}{ds}\mathcal{F}^d(s)\leq0$ when the second term on the right-hand side of (\ref{eq:derivative  dF(s)/ds }) takes the maximal value. 
Hence, by requiring (\ref{eq:derivative  dF(s)/ds }) to be negative when $\Bar{\Theta}=1$, we arrive at (\ref{eq:inequality for dominant strategy}).
\qed
\end{proof}
According to the above result, when the relative reward $r$ satisfies (\ref{eq:inequality for dominant strategy}), a player chooses $s_{\min}$ no matter what she observes from the information broadcast. 
We regard the relative reward satisfying the equality in (\ref{eq:inequality for dominant strategy}) as the critical relative reward and denote it by $r_{\mathrm{crit}}$.
Condition (\ref{eq:inequality for dominant strategy}) is relatively demanding, since it requires the reward perceptions of all players in the populations to go to one extreme, \ie, the reward of socially inactive is low whatever the status of the epidemic is.
However, with $s_{\min}$ being a potential dominant strategy, the information broadcaster can make the condition on $r$ for enabling dominant strategies less restrictive by taking advantage of $\Bar{\Theta}$. 
Specifically, by misreporting $\Bar{\Theta}$ with a value $\Tilde{\Theta}$ satisfying $0\leq\Tilde{\Theta}<\Bar{\Theta}$ in every information broadcast, the information broadcaster makes $s_{\min}$ dominant for all players even if $\abs{r} < \abs{r_{\mathrm{crit}}}$. 
This shows the destructive impact of misinformation. 
Indeed, even when a player possesses a reward perception that social inactivity during the epidemic is acceptable,
if the media consecutively underreport the epidemic, this player would underestimate possible consequences of infection and become highly socially active, \ie, playing strategy $s_{\min}$. 
As a consequence, few infections cause regional outbreaks of the epidemic, and the epidemic eventually becomes a pandemic.

\section{Time-dependent behavior}
\label{sec:time dependent}
In  Section \ref{sec:long term}, we have assumed that the epidemic dynamics evolve at a faster time scale. In this section, we investigate our framework at a different time scale. We first present a result analogous to that of Section \ref{sec:long term:long term property of the game } to make connections with the rate-$\tau$ exponential distribution that controls the times between the information broadcasts. 
Then, we show that the essence of the framework is maintained when we consider the approximate time-dependent epidemic dynamics.  
\par

\subsection{Time-dependent property of the game}
\label{sec:time dependent:short term property of the game}
We emphasize the time-dependence of the probability of infection using $\eta^d_i(\Delta t)$, where $\Delta t$ denotes the time between information broadcasts. We rewrite the payoffs of players as:
\begin{equation}
    F^d_i=s^d_ir-(1-s^d_i)\eta^d_i(\Delta t).
    \label{eq:payoff function time-dependent}
\end{equation}
The next result is the time-dependent counterpart of Theorem \ref{thm:submodular game long term}.
\begin{theorem}
\label{thm:time-dependent submodular game}
Under the assumption that all the players have the same probability to be infected at the beginning, $\ie$ $I^d_i(0)=i_0, \forall i\in\mathcal{I}^d, \forall d\in\mathcal{D}$, the game defined in Section \ref{sec:problem setting: population game} with payoff functions (\ref{eq:payoff function time-dependent}) is a submodular game for any time $\Delta t>0$ between information broadcasts.
\end{theorem}
\begin{proof}
We first show that $I^d_i(t)$ are ordered when $I^d_i(0)=i_0, \forall i\in\mathcal{I}^d, \forall d\in\mathcal{D}$.
From (\ref{eq:SI dynamics}), we observe that the distinct term is $\lambda^d_i=\lambda(1-s^d_i)$. 
For all $i,j\in\mathcal{I}^d$ such that $i>j$ and $s^d_i,s^d_j\in\mathcal{S}^d$, 
since $s^d_i>s^d_j$, $\dot{I}^d_i(t)$ is an upper bound of $\dot{I}^d_j(t)$ for all $d\in\mathcal{D}$.
At an arbitrary time $t>0$, we obtain the expression
\begin{equation*}
    I^d_i(t)=I^d_i(0)+\int_0^t -\gamma I^d_i(t)+\lambda (1-s^d_i)(1-I^d_i(t))d\Theta(t) dt.
\end{equation*}
Since $I^d_i(0)=i_0, \forall i\in\mathcal{I}^d, \forall d\in\mathcal{D}$, we conclude that $I^d_j(t)>I^d_i(t)$ for all $t>0$ and $d\in\mathcal{D}$ if $i>j$ using arguments in \cite{brauer1963bounds}. 
Note that the social state $x$ only appears in the expression of $\Theta(t)$ and $\Theta(t)$ stays the same for all $s^d_i\in\mathcal{S}^d$ and $d\in\mathcal{D}$. Therefore, the relation $I^d_j(t)>I^d_i(t)$ holds when the social state $x$ evolves.
\par
Next, we derive the structural properties of the matrix $DF^d_c(x)$ with the payoff (\ref{eq:payoff function time-dependent}). We focus on the time period $[t,t+\Delta t]$. The evolution of the epidemic dynamics is 
\begin{equation}
    I^d_i(t+\Delta t)=I^d_i(t)+\int_t^{t+\Delta t} -\gamma I^d_i(t)+\lambda (1-s^d_i)(1-I^d_i(t))d\Theta(t) dt.
    \label{eq:integral on [t,t+dt]}
\end{equation}
The only term that depends on the social state in (\ref{eq:integral on [t,t+dt]}) is $\Theta(t)$.
Hence, the derivative of the payoff (\ref{eq:payoff function time-dependent}) with respect to the social state is
\begin{equation}
    \frac{\partial F^d_i}{\partial x^c_j}
    =-\frac{\lambda}{\bar{d}}dc\int_t^{t+\Delta t}
    [(1-s^d_i)^2(1-I^d_i(t))][I^c_j(t)]dt. 
    \label{eq:DF dicj}
\end{equation}
The matrix $DF^d_c(x)$ takes (\ref{eq:DF dicj}) as the element on the row corresponding to strategy $s^d_i\in\mathcal{S}^d$ of population $d$ and the column corresponding to strategy $s^c_j\in\mathcal{S}^c$ of population $c$.
Our objective is to show that $(\Sigma^d)^T DF^d_c(x) \Sigma^c\leq 0$ for all $d,c\in\mathcal{D}$.
We arrive at the expression of the  element in $(\Sigma^d)^T DF^d_c(x) \Sigma^c$ on row $i$ column $j$ as follows:
\begin{equation}
    [(\Sigma^d)^T DF^d_c(x) \Sigma^c]^{d,c}_{i,j}
    =-\frac{\lambda dc}{\bar{d}}(\Upsilon^{d,c}_{i+1,j+1}-\Upsilon^{d,c}_{i,j+1}-\Upsilon^{d,c}_{i+1,j}+\Upsilon^{d,c}_{i,j}),
\end{equation}
where $\Upsilon^{d,c}_{i,j}=\int_t^{t+\Delta t}[(1-s^d_i)^2(1-I^d_i(t))][I^c_j(t)]dt$.
By combining terms, we obtain
\begin{equation}
\begin{aligned}
    & \Upsilon^{d,c}_{i+1,j+1}-\Upsilon^{d,c}_{i,j+1}-\Upsilon^{d,c}_{i+1,j}+\Upsilon^{d,c}_{i,j} \\
    = 
    &\int_t^{t+\Delta t}[I^c_{j+1}(t)-I^c_j(t)][(1-s^d_{i+1})^2(1-I^d_{i+1}(t))-(1-s^d_i)^2(1-I^d_i(t))]dt.
    \label{eq:Upsilon term}
\end{aligned}
\end{equation}
Then, it suffices to prove that (\ref{eq:Upsilon term}) is nonnegative. 
\par
A sufficient condition for (\ref{eq:Upsilon term}) to be nonnegative is that the integrand in (\ref{eq:Upsilon term}) is nonnegative on the interval $[t,t+\Delta t]$. We have shown that $I^d_j(t)>I^d_i(t)$ when $i>j$. 
Hence, it suffices to show that $[(1-s^d_{i+1})^2(1-I^d_{i+1}(t))-(1-s^d_i)^2(1-I^d_i(t))]$ is negative.
By rearranging the differential equation (\ref{eq:SI dynamics}), we arrive at
\begin{equation}
    \frac{\dot{I}^d_i(t)+\gamma I^d_i(t)}{\lambda d \Theta(t)}=(1-s^d_i)(1-I^d_i(t)).
    \label{eq:rearrage ODE}
\end{equation}
Multiplying $1-s^d_i$ on both sides of (\ref{eq:rearrage ODE}), we obtain
\begin{equation}
    \frac{\dot{I}^d_i(t)+\gamma I^d_i(t)}{\lambda d \Theta(t)}(1-s^d_i)=(1-s^d_i)^2(1-I^d_i(t)).
    \label{eq:rearrage ODE and multiply 1-s}
\end{equation}
By (\ref{eq:rearrage ODE and multiply 1-s}), it suffices to prove that the left-hand side of (\ref{eq:rearrage ODE and multiply 1-s}) is decreasing with respect to $i$. 
We have shown that $\dot{I}^d_j(t)$ upper bounds $\dot{I}^d_i(t)$ and $I^d_j(t)>I^d_i(t)$ if $i>j$. 
In addition, the strategies in the set $\mathcal{S}^d$ follow an increasing order.
Therefore, we conclude that the left-hand side of (\ref{eq:rearrage ODE and multiply 1-s}) decreases as $i$ increases.
This completes the proof.
\qed
\end{proof}

We remark that in the proof of Theorem \ref{thm:time-dependent submodular game}, the time $\Delta t$ is arbitrary. 
This suggests the possibility of different time intervals between two information broadcasts. 
Hence, we can assume that the times between the information broadcasts are independent and follow a rate $\tau$ exponential distribution, which coincides with the settings described in Section \ref{sec:problem setting}. 
\par
With Theorem \ref{thm:time-dependent submodular game}, the convergence result in Corollary \ref{thm:BRD convergence} can be extended to the setting where the epidemic evolves for an arbitrary time between any two information broadcasts. Indeed, after the learning procedure defined in (\ref{eq:BRD}) does not result in new social states given additional information broadcasts, the epidemic will gradually converge to the unique positive steady state associated with the current social state as $t\rightarrow \infty$.
This leads to Nash equilibria defined in Definition \ref{def:Nash equilibrium}.
\par

\subsection{Equivalent networks and approximations}
\label{sec:time dependent:equivalent small network}
In general, the time-dependent behaviors depend on the solution of a system of nonlinear differential equations. Here, we study it using approximations.
\par
Combining (\ref{eq:SI dynamics}) and (\ref{eq:Theta}), we obtain
\begin{equation}
    \begin{aligned}
        \frac{d}{dt}I^d_i(t)&=-\gamma I^d_i(t)+\lambda d(1-s^d_i)\left[ 1-I^d_i(t) \right] \frac{\sum_{d\in\mathcal{D}}\sum_{i\in\mathcal{I}^d}dx^d_iI^d_i(t)}{\sum_{d\in\mathcal{D}}dm^d}\\
        &=-\gamma I^d_i(t)+\lambda \left[ 1-I^d_i(t) \right]\frac{d(1-s^d_i)\sum_{d\in\mathcal{D}}\sum_{i\in\mathcal{I}^d}dx^d_iI^d_i(t)}{\Bar{d}}\\
        &=-\gamma I^d_i(t)+\lambda \left[ 1-I^d_i(t) \right]\sum_{d'\in\mathcal{D},i'\in\mathcal{I}^{d'}}\Tilde{A}_{(d,i),(d',i')}I^{d'}_{i'}(t),
    \end{aligned}
    \label{eq:before linearization}
\end{equation}
where $\Tilde{A}_{(d,i),(d',i')}=\frac{d(1-s^d_i)d'x^{d'}_{i'}}{\Bar{d}}$ denotes an entry in the matrix $\Tilde{A}\in\mathbb{R}^{n\times n}$, with $(d,i)$ representing the row index and $(d',i')$ representing the column index. We use a pair $(d,i)$ to represent an index  to emphasize that this index is associated with population $d\in\mathcal{D}$ and strategy $s^d_i\in\mathcal{S}^d$, $i\in\mathcal{I}^d$.
The matrix $\Tilde{A}$ acts as an equivalent adjacency matrix if we regard our epidemic dynamics as a dynamical system on a small network with $n$ nodes. 
The entry $\Tilde{A}_{(d,i),(d',i')}$ stands for the weight on the link from node $(d,i)$ to node $(d',i')$. In general, matrix $\Tilde{A}$ is asymmetric, indicating that the equivalent network is directed. Since we have interpreted the population game as a $D$-player game in Section \ref{sec:long term NE}, the interactions among the $D$ players can be captured by the small network. 
\par
To analyze the time-dependent behavior of the epidemics, we ignore the quadratic terms in the equivalent dynamics and arrive at:
\begin{equation*}
    \frac{d}{dt}I^d_i(t)\simeq-\gamma I^d_i(t)+\lambda \sum_{d'\in\mathcal{D},i'\in\mathcal{I}^{d'}}\Tilde{A}_{(d,i),(d',i')}I^{d'}_{i'}(t).
\end{equation*}
Combining the two terms and rewriting it in matrix form, we obtain
\begin{equation}
    \frac{d}{dt}I(t)\simeq\lambda AI(t),
    \label{eq:linearization by igoring quadratic terms}
\end{equation}
where $A=\Tilde{A}-\frac{\gamma}{\lambda}\Lambda \in\mathbb{R}^{n\times n} $ and $\Lambda$ represents the identity matrix. Let $v_k$ and $\kappa_k$ denote the $k$-th eigenvector and eigenvalue of matrix $A$.
Then, we can express $I(t)$ as a linear combination of the eigenvectors as
\begin{equation*}
    I(t)\simeq\sum_{k=1}^n \alpha_k(t)v_k,
\end{equation*}
where $\alpha_k(t)$ is a time-dependent parameter associated with $v_k$. To solve for $\alpha_k(t)$, we use the following equation:
\begin{equation*}
    \sum_{k=1}^n\frac{d}{dt}\alpha_k(t)v_k=\frac{d}{dt}I(t)=\lambda A I(t)=\lambda A \sum_{k=1}^n\alpha_k(t)v_k=\lambda\sum_{k=1}^n\kappa_k\alpha_k(t)v_k.
\end{equation*}
Hence,
\begin{equation}
    \frac{d}{dt}\alpha_k(t)=\lambda\kappa_k\alpha_k(t).
    \label{eq:ODE to solve alpha}
\end{equation}
The differential equation (\ref{eq:ODE to solve alpha}) has solutions $\alpha_k(t)=\alpha_k(0)e^{\lambda \kappa_k t}$. Therefore, we obtain the time-dependent behavior of the epidemic:
\begin{equation}
    I(t)\simeq\sum_{k=1}^n\alpha_k(0)e^{\lambda \kappa_k t}v_k.
    \label{eq:time dependent behavior of I}
\end{equation}
The exponential term determines the growth speed of $I(t)$. Hence, the largest eigenvalue $\kappa_1$ corresponds to the fastest evolution component. 
\par
In general, (\ref{eq:time dependent behavior of I}) can be approximated using $I(t)\simeq \alpha_1(0)e^{\lambda \kappa_1 t}v_1$ by assuming that the largest eigenvalue is significantly greater than the second largest eigenvalue. 
In our case, we can leverage the structural properties of matrix $\Tilde{A}$ to justify the approximation given by
\begin{equation}
    I(t)\simeq\alpha_1(0)e^{\lambda \kappa_1 t}v_1.
    \label{eq:largest eigenvalue approxiamtion}
\end{equation}
From the definition of $\Tilde{A}_{(d,i),(d',i')}$, we obtain
\begin{equation*}
    \Tilde{A}= \frac{1}{\Bar{d}}
    \begin{pmatrix}
    1\cdot (1-s^1_1) \\
    \vdots \\
    D(1-s^D_{n^D})
    \end{pmatrix}
    \cdot 
    \begin{pmatrix}
    1\cdot x^1_1 & \cdots & Dx^D_{n^D} 
    \end{pmatrix}.
\end{equation*}
This shows that matrix $\Tilde{A}$ has rank $1$. Furthermore, $\Tilde{A}$ has a nonnegative eigenvalue $\Tilde{\kappa}_1=\bar{d}^{-1}\sum_{d\in\mathcal{D}}\sum_{i\in\mathcal{I}^d}d^2(1-s^d_i)x^d_i$ of order $1$ associated with the eigenvector $\Tilde{v}_1=\begin{pmatrix}
1\cdot (1-s^1_1) &
    \cdots &
    D(1-s^D_{n^D})
\end{pmatrix}^T$, and an eigenvalue $0$ of order $n-1$ associated with the eigenvectors $\Tilde{v}_k\in Null(\Tilde{A})$ if $k\neq 1$.
Since $A=\Tilde{A}-\frac{\gamma}{\lambda}\Lambda$, we arrive at $\kappa_1=\Tilde{\kappa}_1-\frac{\gamma}{\lambda}$ associated with $v_1=\Tilde{v}_1$, and $\kappa_k=-\frac{\gamma}{\lambda}$ associated with $v_k=\Tilde{v}_k$ if $k\neq 1$. 
Consider the conditions $\frac{\lambda d(1-s^d_i)}{\gamma}\geq 1,\forall i\in\mathcal{I}^d, \forall d\in\mathcal{D}$ which we have assume in Theorem \ref{thm:epidemic GAS} to obtain the positive steady state. 
Then, we obtain $\Tilde{\kappa}_1\geq \frac{1}{\bar{d}}\frac{\gamma}{\lambda}\sum_{d\in\mathcal{D}}\sum_{i\in\mathcal{I}^d}dx^d_i=\frac{\gamma}{\lambda}$. 
Hence, $\kappa_1\geq 0$. This shows that the largest eigenvalue is nonnegative when $\frac{\lambda d (1-s^d_i)}{\gamma}\geq 1$. 
In addition, we know that $\kappa_k=-\frac{\gamma}{\lambda}<0$ for all $k\neq1$. 
Therefore, the initial values $\alpha_k(0)$ in (\ref{eq:time dependent behavior of I}) decays exponentially in time if $k\neq 1$. 
This justifies the approximation (\ref{eq:largest eigenvalue approxiamtion}). 
\par
Some of the properties of our framework introduced in the previous sections depends on the submodularity of the game. The approximation (\ref{eq:largest eigenvalue approxiamtion}) does not break this structural feature. Consider the derivative of the payoff under the approximation (\ref{eq:largest eigenvalue approxiamtion}):
\begin{equation}
    \frac{\partial}{\partial x^c_j}F^d_i(x,t)=-\alpha_1(0)\lambda t e^{\lambda \kappa_1 t}d(1-s^d_i)^2c^2(1-s^c_j).
    \label{eq:DF approximate}
\end{equation}
In (\ref{eq:DF approximate}), both $d(1-z)^2$ and $c^2(1-z)$ are decreasing functions in $z$ on $z\in[0,1]$. Hence, we can show that the inequality   $(\Sigma^d)^TDF^d_c(x)\Sigma^c\leq 0$ holds when $DF^d_c(x)$ is based on (\ref{eq:DF approximate}). Therefore, after the approximation, players' decisions remain strategic substitutes.

\par
\subsection{Carleman linearization }
\label{sec:time dependent:Carleman}
\par
In the above analysis based on approximation, we have ignored the influence of the quadratic term. We now introduce another technique, called Carleman linearization \cite{steeb1980non}. This approximation method considers an iterative procedure which linearizes a nonlinear system with increasing accuracy as the iterations proceed. We provide the explicit linearization construction procedure.
\par
We first reformulate (\ref{eq:before linearization}) as
\begin{equation}
\begin{aligned}
    \dot{I^d_i}&=\lambda\sum_{d',i'}A_{(d,i),(d',i')}I^{d'}_{i'}-\lambda\sum_{d',i'}\Tilde{A}_{(d,i),(d',i')}I^d_iI^{d'}_{i'} \\
    &=\lambda\sum_{d',i'}A_{(d,i),(d',i')}I^{d'}_{i'}-\lambda\sum_{d',i',d'',i''}B_{(d,i),(d',i'),(d'',i'')}I^{d'}_{i'}I^{d''}_{i''},
    \label{eq:reformulation before Carleman}
\end{aligned}
\end{equation}
where $B_{(d,i),(d',i'),(d'',i'')}=\Tilde{A}_{(d,i),(d',i')}$ if $(d,i)=(d'',i'')$ and $B_{(d,i),(d',i'),(d'',i'')}=0$ otherwise. Let $\otimes$ denote the Kronecker product of two matrices. Define $I^{dd'}_{ii'}:=I^d_i I^{d'}_{i'}$. Let
$
    I_{[2]}=\begin{pmatrix}
    I^{11}_{11} & \cdots I^{1D}_{1n^D} &\cdots I^{D1}_{n^D1} & \cdots & I^{DD}_{n^D n^D} 
    \end{pmatrix}^T
$
denote the concatenation of the quadratic terms of $I$. 
Define matrix $B\in\mathbb{R}^{n\times^2}$ by
\begin{equation*}
    B=\begin{pmatrix}
    B_{(1,1),(1,1),(1,1)} & B_{(1,1),(1,1),(1,2)} & \cdots
    & B_{(1,1),(D,n^D),(D,n^D)} \\
    \vdots & \vdots & \ddots & \vdots \\
    B_{(D,n^D),(1,1),(1,1)} & B_{(D,n^D),(1,1),(1,2)} & \cdots & B_{(D,n^D),(D,n^D),(D,n^D)}
    \end{pmatrix}.
\end{equation*}
We use $\Lambda_l$ to denote the identity matrix of size $l$.
Then, we express (\ref{eq:reformulation before Carleman}) as
\small
\begin{equation}
\begin{aligned}
    &
    \begin{pmatrix}
    \dot{I} \\
    \dot{I}_{[2]}
    \end{pmatrix}
    =\lambda 
    \begin{pmatrix}
    A & -B \\
    0 & A\otimes\Lambda_{n} + \Lambda_n \otimes A
    \end{pmatrix}
    \begin{pmatrix}
    I \\
    I_{[2]}
    \end{pmatrix}
    - \lambda \cdot \\
    &
    \begin{pmatrix}
    0 \\
    0 \\
    \sum\limits_{\substack{d',i',d'',i''}}B_{(1,1),(d',i'),(d'',i'')}I^1_1I^{d'}_{i'}I^{d''}_{i''}+\sum\limits_{\substack{d',i',d'',i''}}B_{(1,1),(d',i'),(d'',i'')}I^{d'}_{i'}I^{d''}_{i''}I^{1}_{1} \\
    \vdots \\
    \sum\limits_{\substack{d',i',d'',i''}}B_{(D,n^D),(d',i'),(d'',i'')}I^{D}_{n^D}I^{d'}_{i'}I^{d''}_{i''}+\sum\limits_{\substack{d',i',d'',i''}}B_{(D,n^D),(d',i'),(d'',i'')}I^{d'}_{i'}I^{d''}_{i''}I^{D}_{n^D} 
    \end{pmatrix}
    .
    \label{eq:2 order Carleman}
\end{aligned}
\end{equation}
\normalsize
As described in \cite{steeb1980non}, by iteratively defining $A_{[2]}=A\otimes\Lambda_{n} + \Lambda_n \otimes A$ and $B_{[2]}=B\otimes\Lambda_{n} + \Lambda_n \otimes B$ until
$A_{[C]}=A\otimes \Lambda_{n^{C-1}}+\Lambda_n \otimes A_{[C-1]}$ and $B_{[C]}=B\otimes \Lambda_{n^{C-1}}+\Lambda_n \otimes B_{[C-1]}$, we obtain the order-$C$ Carleman linearization of system (\ref{eq:reformulation before Carleman}) as follows:
\begin{equation}
    \begin{pmatrix}
    \dot{I} \\
    \dot{I}_{[2]} \\
    \vdots \\
    \dot{I}_{[C-1]} \\
    \dot{I}_{[C]}
    \end{pmatrix}
    \simeq
    \lambda
    \begin{pmatrix}
    A & -B & 0 & \cdots & 0 & 0 \\
    0 & A_{[2]} & -B_{[2]} & \cdots & 0 & 0 \\
    \vdots & \vdots & \vdots & \ddots & \vdots & \vdots \\
    0 & 0 & 0 & \cdots & A_{[C-1]} & -B_{[C-1]} \\
    0 & 0 & 0 & \cdots & 0 & A_{[C]}
    \end{pmatrix}
    \begin{pmatrix}
    I \\
    I_{[2]}\\
    \vdots \\
    I_{[C-1]} \\
    I_{[C]}
    \end{pmatrix}.
    \label{eq: C order Carleman linearization}
\end{equation}
Let $\rho(\omega)$ denote the Laplace transform of $I(t)$. 
Then, we can solve the system of linear differential equations in (\ref{eq: C order Carleman linearization}) by solving the linear system
\begin{equation}
    \lambda
    \begin{pmatrix}
    \Delta & -B & 0 & \cdots & 0 & 0 \\
    0 & \Delta_{[2]} & -B_{[2]} & \cdots & 0 & 0 \\
    \vdots & \vdots & \vdots & \ddots & \vdots & \vdots \\
    0 & 0 & 0 & \cdots & \Delta_{[C-1]} & -B_{[C-1]} \\
    0 & 0 & 0 & \cdots & 0 & \Delta_{[C]}
    \end{pmatrix}
    \begin{pmatrix}
    \rho \\
    \rho_{[2]}\\
    \vdots \\
    \rho_{[C-1]} \\
    \rho_{[C]}
    \end{pmatrix}
    =-
    \begin{pmatrix}
    I(0) \\
    I_{[2]}(0)\\
    \vdots \\
    I_{[C-1]}(0) \\
    I_{[C]}(0)
    \end{pmatrix},
    \label{eq:linear system to solve for laplace variables}
\end{equation}
where $\Delta=A-\omega \Lambda_n$, $\Delta_{[2]}=A_{[2]}-\omega \Lambda_{n^2}$ and so on. The inverse of the matrix on the left-hand side of (\ref{eq:linear system to solve for laplace variables}) has a closed form according to \cite{steeb1980non}. Hence, we can find $\rho, \rho_{[2]},...,\rho_{[C]}$ and take the inverse Laplace transform to find $I,I_{[2]},...,I_{[C]}$. The accuracy of the linearization depends on the value of $C$. For numerical computations, the approximation is stopped at a sufficiently large $C$ under which the Carleman linearization is sufficiently accurate.
\par
The Carleman approximation (\ref{eq: C order Carleman linearization}) is more accurate than the approximation (\ref{eq:linearization by igoring quadratic terms}) for estimating the time-dependent growth of the original nonlinear system (\ref{eq:SI dynamics}). 
The explicit forms of the matrices $A,...,A_{[C]}$ and $B,...,B_{[C-1]}$ in (\ref{eq: C order Carleman linearization}) make it convenient to implement the iterative Carleman linearization procedures. 
After determining the order $C$, the Carleman linearization neglect high-order terms with the subscripts higher than $[C]$, which introduce errors to the computation of $I$ through the extended states $I_{[2]},...,I_{[C]}$ of the system (\ref{eq: C order Carleman linearization}).
However, according to \cite{hashemian2019feedback}, we can compensate this drawback by periodically resetting extended states and recomputing them from the original state $I$. 
In these resettings, there is no need to recompute the matrices $A,...,A_{[C]}$ and $B,...,B_{[C-1]}$.

As pointed out in \cite{hashemian2019feedback,fang2016formulation}, the equations of the sensitivity analysis of the states leads to the same forms in both linear and nonlinear systems approximated by the Carleman linearization. 
This suggests the possibility that the gradients of the payoffs (\ref{eq:payoff function time-dependent}) under  (\ref{eq: C order Carleman linearization}) share the same analytical properties as (\ref{eq:linearization by igoring quadratic terms}).
Specifically, we observe that the matrix $A_{[C]}$ is constructed using a Kronecker sum \cite{horn2012matrix}, whose eigenvalues are the sums of the eigenvalues of its components. 
The components $A\otimes \Lambda_{n^{C-1}}$ and $\Lambda_n\otimes A_{[C-1]}$ are themselves Kronecker products, whose eigenvalues are products of eigenvalues of their factors. 
Hence, the eigenvalues of $A_{[C]}$ are the eigenvalues of $A$ multiplied by the scalar $C$.
Therefore, by applying similar techniques as discussed in Section \ref{sec:time dependent:equivalent small network} to (\ref{eq: C order Carleman linearization}), we obtain
\begin{equation}
    \begin{pmatrix}
    \dot{I} \\
    \dot{I}_{[2]} \\
    \vdots \\
    \dot{I}_{[C-1]} \\
    \dot{I}_{[C]}
    \end{pmatrix}
    \simeq
    \alpha_{1,[C]}(0)e^{\lambda \kappa_{1,[C]}t}v_{1,[C]},
    \label{eq:largest eigenvalue approxiamtion Carleman}
\end{equation}
where $\kappa_{1,[C]}$ and $v_{1,[C]}$ represent the largest eigenvalue and the corresponding eigenvector of the system matrix of (\ref{eq: C order Carleman linearization}), and $\alpha_{1,[C]}(0)$ is the initial value of the time-dependent parameter, analogous to (\ref{eq:ODE to solve alpha}), satisfying $\frac{d}{dt}\alpha_{1,[C]}(t)=\lambda \kappa_{1,[C]}\alpha_{1,[C]}(t)$.
The time-dependent behavior of $I$ is obtained from (\ref{eq:largest eigenvalue approxiamtion Carleman}) by selecting the related terms. 
Leveraging these terms, we observe, by following similar arguments as those for (\ref{eq:DF approximate}), that players' decisions remain strategic substitutes for time-dependent epidemic behaviors under the Carleman linearization in our game framework.  In addition, (\ref{eq: C order Carleman linearization}) provides a way to estimate herd behaviors involved in $A,...,A_{[C]}$ using empirical epidemic data $I,\cdots,I_{[C]}$ at different time instances. 

\section{Numerical experiments}
\label{sec:numerical experiments}

In this section, we present the results of numerical experiments. 
For presentation purposes, we set the number of populations $D$ to be $2$, with numbers of strategies $n^1$ and $n^2$ being $3$ and $2$, respectively. 

\subsection{Convergence results}
\label{sec:numerical experiments:convergence}

\begin{figure}[ht]
    \centering
    \subfigure[]{\includegraphics[width=0.465\textwidth]{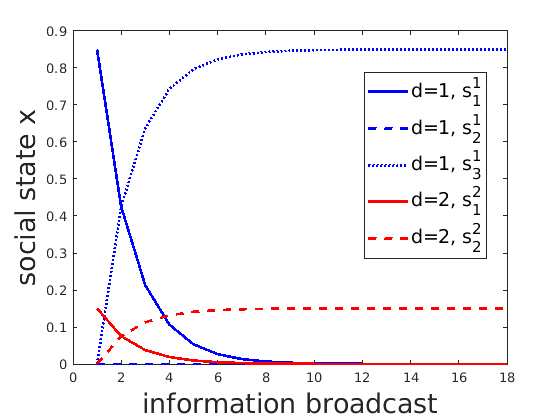}
    \label{fig:subfig1.1}
    }
    \subfigure[]{\includegraphics[width=0.487\textwidth]{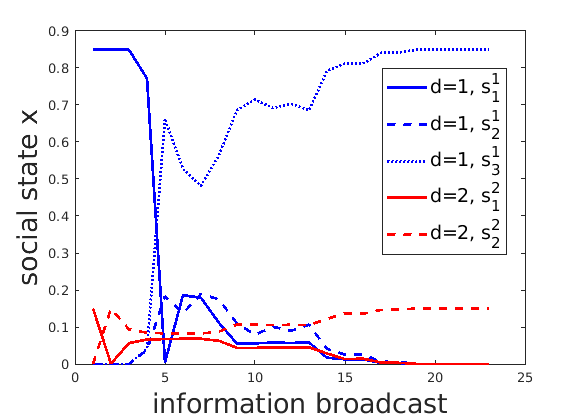}
    \label{fig:subfig1.2}
    }
    \caption[Optional caption for list of figures]{Convergence to the NE: (a) Using best-response dynamics (\ref{eq:BRD}). (b) Solving the optimization problem (\ref{eq:optimization prob}).} 
    \label{fig:fig1}
\end{figure}

\begin{figure}[ht]
    \centering
    \subfigure[]{\includegraphics[width=0.46\textwidth]{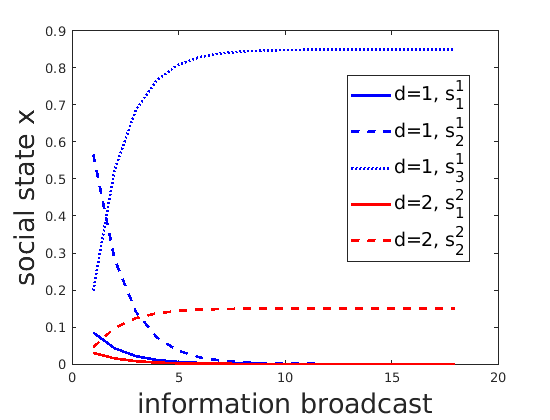}
        \label{fig:subfig2.1}
    }
    \subfigure[]{\includegraphics[width=0.472\textwidth]{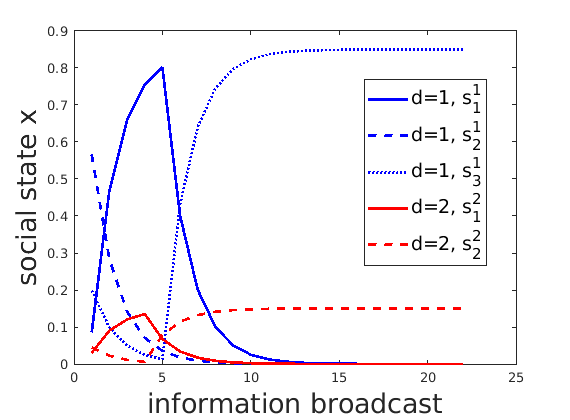}
        \label{fig:subfig2.2}
    }
    \caption[Optional caption for list of figures]{Study of the convergences of best-response dynamics under different scenarios: (a) Information broadcasts take place at the steady states of the epidemic dynamics. (b) Information broadcasts take place before the epidemic dynamics reach the steady states under the linearly approximated epidemic model.} 
    \label{fig:fig2}
\end{figure}

\par
In Fig.\ref{fig:fig1}, the best-response dynamics (\ref{eq:BRD}) and the iterations of the optimization problem (\ref{eq:optimization prob}) converge to the same NE.
According to the evolution of the social state, the best-response dynamics provide a smoother learning process where populations gradually learn the equilibrium strategies through sequential strategic interactions. 
This is due to the monotonicity of the updates of the social states in (\ref{eq:BRD}) under the submodular property of the game.
The learning process provided by the optimization problem (\ref{eq:optimization prob}) is not as well-shaped as the one given by the best-response dynamics.
However, the optimization problem provides a different learning approach containing a global objective for all players, differing from merely considering individual player's myopic reactions toward the payoff realizations.
We interpret the optimization of the problem (\ref{eq:optimization prob}) as an objective-guided evolutionary process describing strategy revisions in the population, with the explicit gradients (\ref{eq:gradient of F}) acting as the incentives of the players in revising their strategies.

In Fig.\ref{fig:fig2}, we compare the convergence of the best-response dynamics when information broadcasts take place at the steady states of the epidemic dynamics and when information broadcasts take place before the epidemic dynamics reach the steady states. Interestingly, in Fig.\ref{fig:subfig2.1} and Fig.\ref{fig:subfig2.2}, the evolution of the social state converges to the same NE, despite the different evolution patterns at the beginning. 
Apart from the different time scales used in Fig.\ref{fig:subfig2.1} and Fig.\ref{fig:subfig2.2}, we have approximated the epidemic evolution using only the linear terms of (\ref{eq:SI dynamics}) in Fig.\ref{fig:subfig2.2}. This result corroborates the approximation methods we have discussed in Section \ref{sec:time dependent:equivalent small network}.

\begin{figure}[ht]
    \centering
    \subfigure[]{\includegraphics[width=0.46\textwidth]{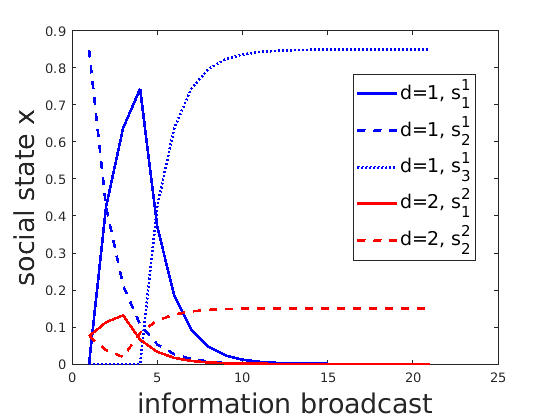}
        \label{fig:subfig3.1}
    }
    \subfigure[]{\includegraphics[width=0.455\textwidth]{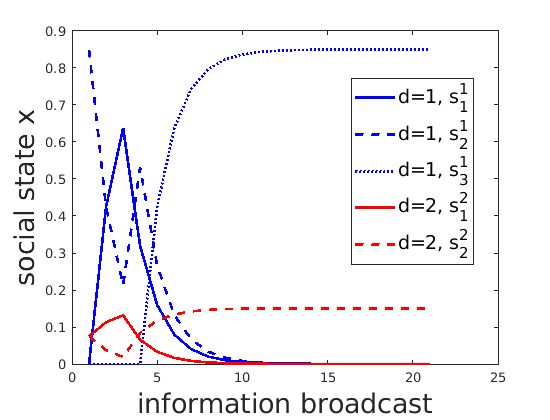}
        \label{fig:subfig3.2}
    }
    \caption[Optional caption for list of figures]{Effect of delaying the information broadcasts. (a) Information broadcasts take place before the epidemic dynamics reach the steady states. (b) Delaying the information broadcasts.} 
    \label{fig:fig3}
\end{figure}

In Fig. \ref{fig:fig3}, we compare the learning processes with and without delays in the information broadcasts. 
As shown in Fig. \ref{fig:fig3}, the delays in the information broadcasts significantly influence the evolution of the behaviors. 
The curve corresponding to $d=1$ and $s^1_2$ is smoother in Fig. \ref{fig:subfig3.1} than in Fig. \ref{fig:subfig3.2}. 
However, both the processes in Fig. \ref{fig:subfig3.1} and in Fig. \ref{fig:subfig3.2} converge to the same equilibrium point.
Delays in information broadcasts often arise from either the delays in data collection and processing or purposeful deferral of the broadcast. 
The results suggest that even if the broadcaster has difficulty in obtaining the epidemic status in real time, the out-of-date information can still lead the herd to reach equilibria, as long as the reported information is chronological and precise. 
Therefore, our framework has the potential to guide real-world policy developing, since neither information reporting nor behavioral revision of the herd is perfectly on time.



\subsection{Infection waves due to strategy changes}

\begin{figure}[ht]
    \centering
    \subfigure[]{\includegraphics[width=0.78\textwidth]{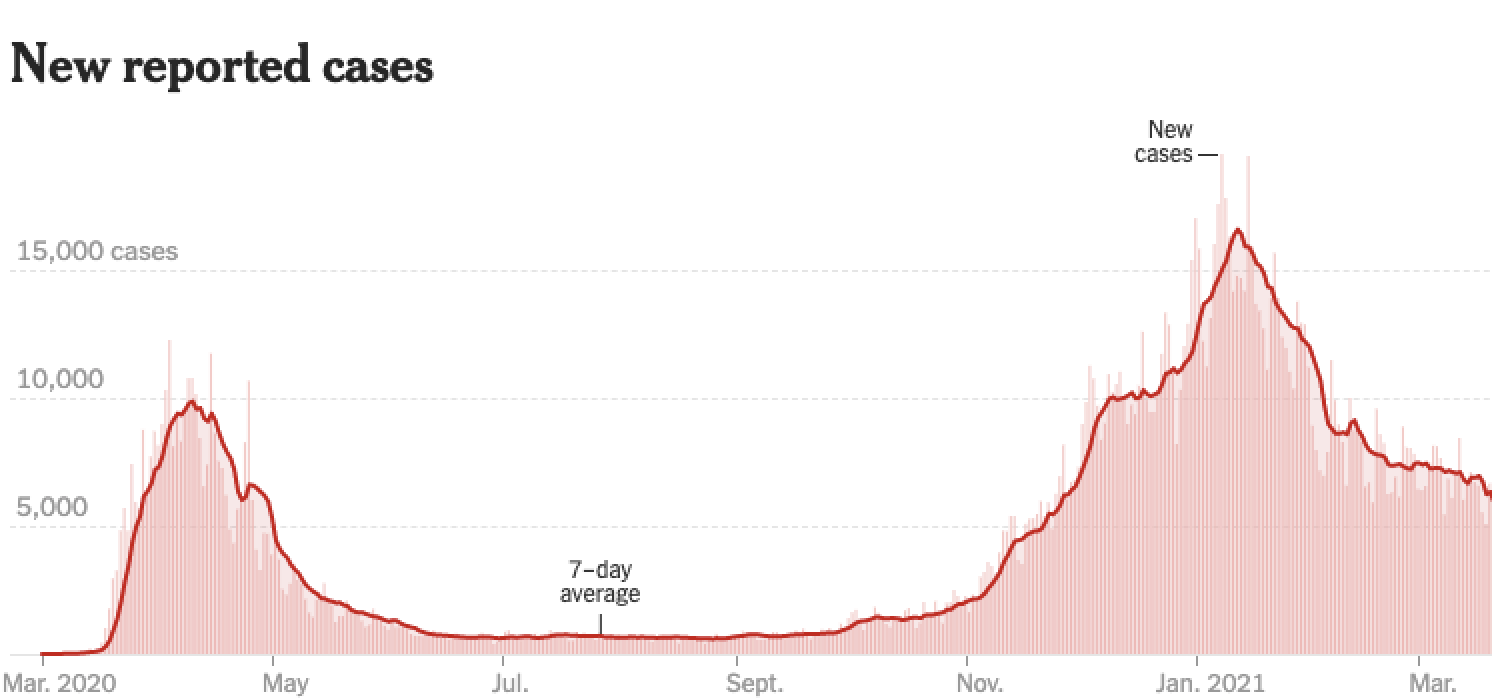}
        \label{fig:subfigreal1}
    }
    \subfigure[]{\includegraphics[width=0.97\textwidth]{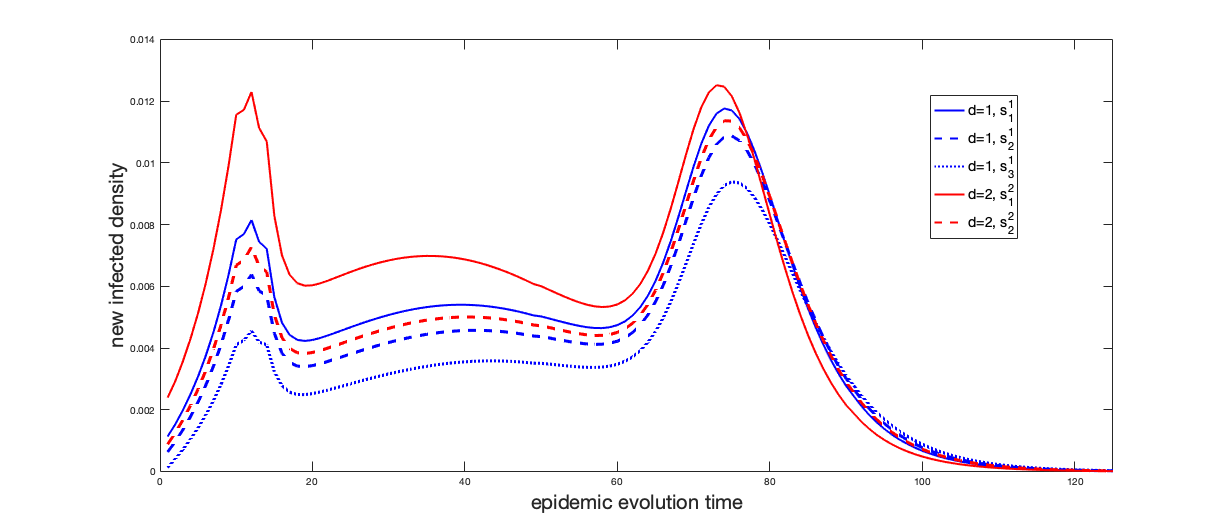}
        \label{fig:subfigreal2}
    }
    \caption[Optional caption for list of figures]{Multiple peaks of new infections. (a) Statistics of COVID-19 in New York \cite{NYTimeswebsite}. (b) Predicted new infected density curve using our framework.} 
    \label{fig:realdata}
\end{figure}
In Fig. \ref{fig:realdata}, we compare the curve of the reported cases using the real COVID-19 data and the simulated curve of the infected density using our framework. 
The multi-peak curves in Fig. \ref{fig:subfigreal1} and Fig. \ref{fig:subfigreal2} correspond to different waves of epidemic outbreak. 
The first wave is the natural outbreak of an epidemic when it first starts to spread among infectious individuals. 
The decreases of new cases between July 2020 and November 2020 in Fig. \ref{fig:subfigreal1} and between time 20 to 50 in Fig. \ref{fig:subfigreal2} correspond to the period when people start to avoid close contacts and the policies are enforced to mitigate the epidemic, such as wearing masks all the time.
In Fig. \ref{fig:subfigreal1}, the second infection wave is a consequence of relaxed social guidance \cite{piller2020undermining} and the violations of existing quarantine policies. 
The behaviors of the populations are set to change at $t=50$ in Fig. \ref{fig:subfigreal2}. This change captures the populations' overconfidence on the epidemic status as the number of new cases decreases.
When the populations become less careful, \ie, more people play strategies close to $s^d_{1}$ in the set $\mathcal{S}^d$, the new infected density curve increases and shows a second peak. 
This second peak is an indication of the influence of the herd behaviors on the epidemic evolution.  

\subsection{Equivalent networks}
\label{sec:numerical experiments:networks}

\begin{figure}[ht]
\centering
\subfigure[]{\includegraphics[width=0.45\textwidth]{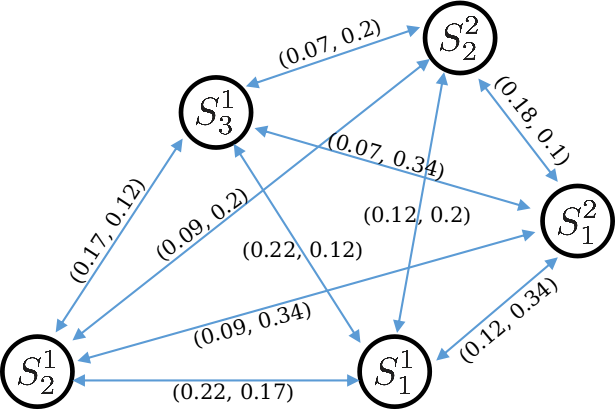}
    \label{fig:subfig4.1}
}
\subfigure[]{\includegraphics[width=0.458\textwidth]{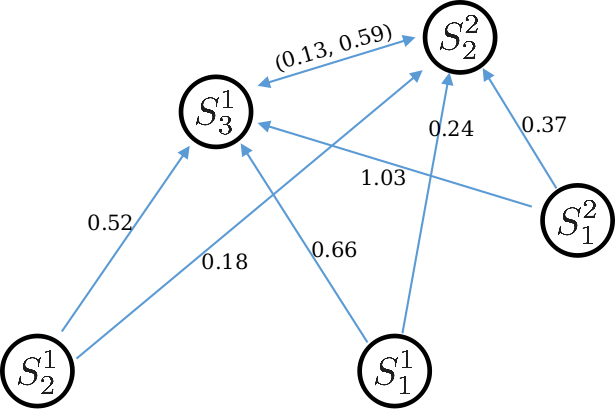}
    \label{fig:subfig4.2}
}
\caption[Optional caption for list of figures]{Equivalent small networks: The networks are directed graphs. Each link is associated with weights. (a) Fully connected initial network. (b) Core-peripheral equilibrium network. } 
\label{fig:fig4}
\end{figure}

Fig. \ref{fig:fig4} illustrates the equivalent small networks discussed in Section \ref{sec:time dependent:equivalent small network}. 
From the randomized initial network in Fig. \ref{fig:subfig4.1}, we eventually arrive at the equilibrium network in Fig. \ref{fig:subfig4.2}, which has a  special structure, where a few nodes are in the center with others being peripheral nodes. 
This pattern is often called a core-peripheral structure \cite{galeotti2010law}. 
In Fig. \ref{fig:subfig4.2}, the nodes $s^1_3$ and $s^2_2$ act as the cores, since all the other nodes have a link pointing towards them but they only have links pointing toward each other. 
Core-peripheral networks are consequences of network formation games where the directed links represent resource flows \cite{galeotti2010law}. 
This observation suggests that the evolution of populations' strategies can be statistically equivalent to an $N$-person  sequential network formation game. 
Each player represents a population in which the individuals are indistinguishable and choose the same strategy. 

\section{Conclusion}
\label{sec:conclusion}
In this paper, we have proposed an evolutionary game framework that couples the state dynamics of the epidemics with the evolution of strategies to study the herd behaviors of the population over complex networks.
We have designed the mechanism containing physical interactions and information broadcasts to combine coupled state transitions over a complex network and sequential strategy revisions in the populations. 
Taking the epidemic model as a special case, we have found a unique nontrivial steady state when the evolution of the epidemic evolves at a faster time scale.

We have characterized the Nash equilibrium of the game at the steady-state and have explicitly expressed its gradient that facilitates numerical computations and the analysis of structural properties of the game. 
In addition, we have shown that decisions in the game are strategic substitutes under arbitrary time scales. This observation has enabled simple learning processes to reach the equilibrium point. 
We have constructed an equivalent small network to represent the complex network, and  proposed structure-preserving approximation methods that can maintain the strategic substitutes property of the game.
Having applied our framework to study the impact of misinformation on epidemics, we have shown that misreports lead to a high social activity intensity, which can exacerbate the spreading of the infection disease.   

\par
Our numerical examples have indicated the predictive power of our framework by comparing the simulated dynamics to the real COVID-19 statistics.
The multi-peak pattern observed in the case study has shown that the herd behaviors have attributed to multiple outbreaks of an epidemic. 
In addition, the numerical experiments on the delays in the information broadcasts suggest an extension of our framework to scenarios where the populations react to the underlying dynamics with a delay.
We would extend our framework to study generalized compartmental models and the role of information structure in the human behaviors in an epidemic. 

\bibliographystyle{spmpsci}      

%
%


\bibliography{bibliography.bib}
\nocite{*}

\end{document}